\documentclass[10pt,a4paper]{article} 

\usepackage[margin=1in]{geometry}
\newtheorem{lemma}{Lemma}
\newtheorem{theorem}{Theorem}{}
\newtheorem{corollary}{Corollary}{}
{}
\usepackage{url}
\usepackage{enumerate}
\usepackage{graphicx}
\usepackage{braket}
\usepackage{todonotes}
\usepackage{amsmath}
\usepackage{hyperref}
\usepackage[english]{babel}
\usepackage{mathrsfs}
\usepackage{amsfonts}
\usepackage{mathtools}
\usepackage{comment}
\usepackage{longtable}
\usepackage{multicol}
\usepackage{multirow}

\usepackage{wrapfig}
\usepackage{lscape}
\usepackage{rotating}
\usepackage{makecell}
\usepackage{authblk}
\usepackage{subcaption}
\usepackage{algorithm,algorithmic}

\title{Efficient Syndrome Decoder for Heavy Hexagonal QECC\\ via Machine Learning}

\author[1,*]{Debasmita Bhoumik}
\author[2]{Ritajit Majumdar}
\author[2]{Dhiraj Madan}
\author[2]{Dhinakaran Vinayagamurthy}
\author[2]{Shesha Raghunathan}
\author[1,+]{Susmita Sur-Kolay}

\affil[1]{Advanced Computing \& Microelectronics Unit, Indian Statistical Institute, India}
\affil[2]{\textit{IBM Quantum}, IBM India Research Lab}

\affil[*]{debasmita.ria21@gmail.com}
\affil[+]{ssk@isical.ac.in}

\date{}
\begin{document}
\maketitle

\begin{abstract}
Error syndromes for heavy hexagonal code and other topological codes such as surface code have typically been decoded by using Minimum Weight Perfect Matching (MWPM) based methods. Recent advances have shown that topological codes can be efficiently decoded by deploying machine learning (ML) techniques, in particular with neural networks. In this work, we first propose an ML based decoder for heavy hexagonal code and establish its efficiency in terms of the values of threshold and  pseudo-threshold, for various noise models. We show that the proposed ML based decoding method achieves $\sim5 \times$ higher values of threshold than that for MWPM.  
    Next, exploiting the property of subsystem codes, we define gauge equivalence for heavy hexagonal code, by which two distinct errors can belong to the same error class.  A linear search based method is proposed for determining the equivalent error classes. This provides a quadratic reduction in the number of error classes to be considered for both bit flip and phase flip errors, and thus a further improvement of $\sim 14\%$ in the threshold over the basic ML decoder. Lastly, a novel technique based on rank to determine the equivalent error classes is  presented, which is empirically faster than the one based on linear search.
\end{abstract}
\providecommand{\keywords}[1]{\textbf{\textit{Index terms---}} #1}



 \section{ Introduction}

Quantum computers excel over their classical counterparts \cite{Shor:1997:PAP:264393.264406, Grover:1996:FQM:237814.237866, arute2019quantum, montanaro2015quantum} for certain computational problems by attaining substantial speedup, due to the quantum mechanical properties of superposition and entanglement. But quantum states are highly fragile and a slightest unwanted rotation that may occur due to an interaction with the environment can introduce errors in the computation. An $[[n,k,d]]$ quantum error correcting code (QECC) encodes $k > 1$ physical qubits into $n > k$ physical qubits for correcting at most $t = \lfloor \frac{d-1}{2} \rfloor$ errors. Some of the earlier QECCs \cite{PhysRevA.52.R2493,PhysRevLett.77.793,PhysRevLett.77.198} are however burdened with the Nearest Neighbour problem \cite{fowler2012towards, bhoumik2022ml} arising in the quantum hardware to realize a logical qubit. Two physical qubits $i$ and $j$ are said to be nearest neighbours if a 2-qubit operation involving $(i,j)$  is feasible. An operation involving two non-neighbour physical qubits is costly since it requires multiple qubit swap operations, thereby reducing the computational speed and in turn making the system more error-prone. Topological QECCs \cite{dennis2002topological,fowler2009high,wang2011surface,wootton2012high} resolve this problem by making the physical qubits for a logical qubit interact only with their neighbours.

Recently, industry research labs have been shifting towards the hexagonal architecture for their quantum computers.  This architecture has the advantage of reducing the number of distinct frequencies, and thus crosstalk. The surface code  \cite{bravyi1998quantum} structure has been modified to a topological code with a heavy hexagonal structure \cite{chamberland2020topological} in order to become more suitable for these architectures. The heavy hexagonal code \cite{chamberland2020topological} uses a combination of degree-two and degree-three vertices in the topology, and can be considered as a hybrid of a surface code and a Bacon-Shor code \cite{bacon2006operator}. This QECC reduces the distinct number of frequencies required in their realization by introducing more ancilla qubits (termed as flag qubits) for entanglement in the syndrome measurement 
\cite{chamberland2020topological}.

For a distance $d$ QECC, if more than $\lfloor \frac{d-1}{2} \rfloor$ errors occur, then the QECC fails to correct those errors, leading to an incorrect logical state called a  \emph{logical error}. A logical error can occur due to incorrect decoding as well. Logical errors pose a serious threat towards building error-corrected qubits since these remain undetected, and are retained in the logical state of the system. Given a QECC, the goal therefore is to design a decoder which reduces the probability of logical error. The performance of a decoder for a QECC is usually quantified by  the following two probabilities.

\begin{figure}[htb]
    \centering
    \includegraphics[scale = 0.4] {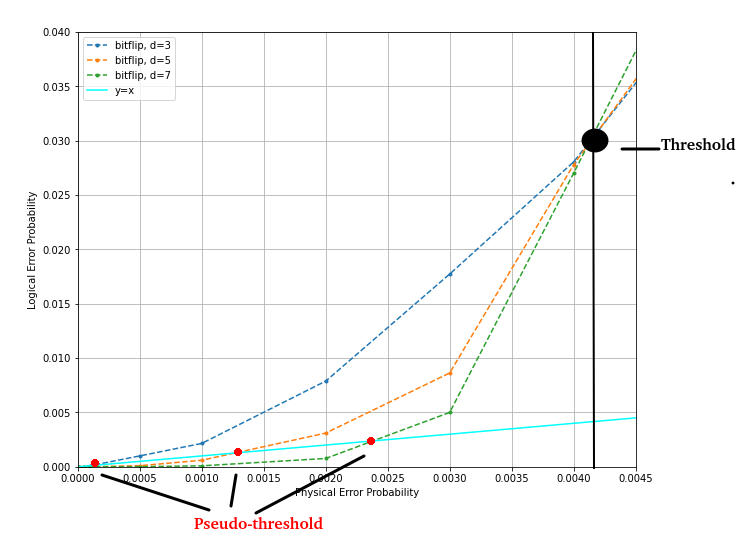}
    \caption{Threshold (black dot) and Pseudo-threshold (red dot) of a MWPM decoder for Heavy Hexagonal  QECC, with performance comparison for distance 3 (\color{blue}blue\color{black}), distance 5 (\color{orange}orange\color{black}), and distance 7 (\color{green}green\color{black}). The \color{cyan} cyan \color{black} straight line $y=x$ is for equal probabilities of physical qubit error and logical error.}
    \label{fig:metric}
\end{figure}

   \emph{Pseudo-threshold} for a QECC is the probability of physical error, below which error correction can effectively reduce the probability of logical error.\\
   
\emph{Threshold} for a QECC is the probability of physical error above which the probability of logical error increases with the distance of the QECC. \\

We have illustrated these two probabilities for heavy hexagonal QECC in Fig. \ref{fig:metric}. A decoder with higher pseudo-threshold and threshold is desirable.
 
\vspace{0.3cm}

\emph{Related works}:\\
The authors of \cite{chamberland2020topological} have proposed the heavy hexagonal code for QECC  and have used an MWPM \cite{edmonds_1965, chamberland2020topological} decoder to evaluate the code. The asymptotic threshold for logical bit flip or $X$ errors is 0.0045. Since phase flip or $Z$ errors are corrected using Bacon-Shor type stabilizers, no threshold for $Z$ errors can be defined as such. To the best of our knowledge, there is no ML based decoder for the heavy hexagonal code, so the threshold 0.0045 is considered the state-of-the-art.

Two widely studied decoders in the literature are the Union Find (UF) \cite{tarjan1975efficiency, wu2022interpretation, delfosse2021almost} and Minimum-weight perfect matching (MWPM) \cite{edmonds_1965, fowler2012towards, chamberland2020topological} decoders. For a graph $G = (V,E)$, Union Find is known to have an almost linear ($\mathcal{O}(|V|)$) decoding time, whereas the decoding time of MWPM scales polynomially as ($\mathcal{O}(|V|^4)$). However, Union Find is shown to have a poorer performance (i.e., lower threshold and pseudo-threshold) than MWPM. Without considering the correlation between the two types of errors \cite{varsamopoulos2017decoding, bhoumik2022ml}, MWPM in turn always aims to find the minimum number of errors ($X$ and $Z$ separately) that can generate an observed syndrome. In \cite{varsamopoulos2017decoding}, the authors established that for surface
code upto distance 7 their feed-forward neural network based decoder   outperforms the MWPM counterpart in the case of depolarizing noise model \cite{nielsen2002quantum}. The authors of  \cite{torlai2017neural} have shown that a Boltzmann machine based decoder for phase flip errors exhibits a logical failure probability that is close to MWPM, but not identical. 

In \cite{bhoumik2022ml}, the authors have proposed a two-level (low and high) ML-based decoding scheme, where the first (low) level corrects errors on the physical qubits and the second (high) level corrects any existing logical errors, for various noise models for surface code. In \cite{azad2022surface}, the authors introduced a [$dx$, $dz$] rectangular surface code design, where $dx$ ($dz$) represents the distance of the code for bit flip (phase flip) error correction, motivated by the fact that the severity of bit flip and phase flip errors in the physical quantum system is asymmetric. They report the values of 
pseudo-threshold and threshold for the proposed surface code design in asymmetric error channels with various degrees of asymmetry of errors in a depolarisation channel.

Machine learning (ML) has been shown to outperform MWPM in terms of time taken and the error thresholds for decoding QECCs such as surface code in various other papers also \cite{chamberland2018deep, krastanov2017deep, sweke2018reinforcement}.

The decoding time of an ML based decoder scales polyonimially ($\mathcal{O}(q^2)$) with the number of qubits $q$ \cite{varsamopoulos2019comparing}. For MWPM and UF decoders, we have  $|V| = \mathcal{O}(q^2)$. Therefore, the decoding via ML decoder is faster than MWPM. Although it is slower than UF, we show in Table~\ref{tab:result1MWPMvsML} that our ML decoder outperforms it significantly in terms of pseudo-threshold and threshold. On the other hand, unlike UF and MWPM, our ML decoder can be trained to take the error probability of the system under consideration during decoding. The ML decoders are trained to map the syndrome bit strings to Pauli error strings. However it is challenging to extend  this ML approach directly to a subsystem code or gauge code such as the heavy hexagonal code.\\

\emph{Our contributions in this article}:\\
This work initiates the study of using ML for decoding heavy hexagonal code. The heavy hexagonal code is a hybrid of surface code and Bacon-Shor code, where the later is a subsystem code \cite{bacon2006operator}. Being a subsystem code, the entire codespace of the heavy hexagonal code is partitioned into equivalent classes (details given in Sec.~\ref{sectiongauge}). By this property, distinct errors can be clubbed into certain equivalent error classes, called gauge equivalence. We identify a unique representative element from each such class. This approach reduces the number of error classes, resulting in a classification problem with fewer classes, and thus the training of ML model becomes faster and more accurate.

For the depolarization noise model \cite{nielsen2002quantum}, we use the entire syndrome (i.e., the syndromes for both $X$ and $Z$ stabilizers) to train the ML decoder even for determining the probability of logical $X$ or $Z$ errors individually. We show by simulation that the na\"ive ML decoder itself achieves a threshold of 0.0137 for logical $X$ errors in bit flip noise, which is much higher than that for the MWPM decoder \cite{chamberland2020topological}. This is further improved to a threshold of 0.0158 using gauge equivalence. We also show that our ML decoder achieves a threshold of 0.0245 for logical $X$ errors in depolarizing noise model, which is better than that for the MWPM decoder \cite{chamberland2020topological}. Similar improvements are observed for phase flip errors as well. In this work we propose the application of classical machine learning for quantum  error decoding to facilitate error correction. Quantum machine learning has not been employed in the work presented here. 

The primary contributions of this article can be summarized as:
\begin{enumerate}
    \item  A new machine learning (ML) based decoder for heavy hexagonal code in bit flip, phase flip and depolarization noise models has been proposed, where the probability of error on each of the eleven steps of this QECC cycle is assumed to be equal \cite{chamberland2020topological}. Our ML  decoder achieves a higher threshold  than that of the MWPM based decoder \cite{chamberland2020topological, higgott2022pymatching}. 
    
    \item An equivalence relation among errors with respect to syndromes has been defined and two algorithms have been designed and implemented based on \emph{linear search} and \emph{rank} to compute a representative from each equivalence class of errors  in order to accelerate syndrome decoding. This  provides a quadratic reduction in the number of classes for both bit flip and phase flip errors. We obtain an empirical improvement in the performance of the ML decoder using this  concept of equivalence for errors.

\end{enumerate}

The rest of this article is organized as follows: Section 2 has the background of heavy hexagonal code structure and ML based decoders. In Section 3, we describe the design methodology of the ML based decoder for heavy hexagonal code. Section 4 first presents the notion of gauge operators \cite{chamberland2020topological}, and introduce their equivalence relations. Using this relation termed as gauge equivalence, we describe our methods for determining the error classes. In Section 5, we discuss the simulation results and  conclude in Section 6.

\section{Background}

Let us now describe the heavy hexagonal code structure, the noise model used in this work, and the motivation for using machine learning based decoder.

 \begin{figure}[htb]
    \centering
    \includegraphics[scale = 0.5] {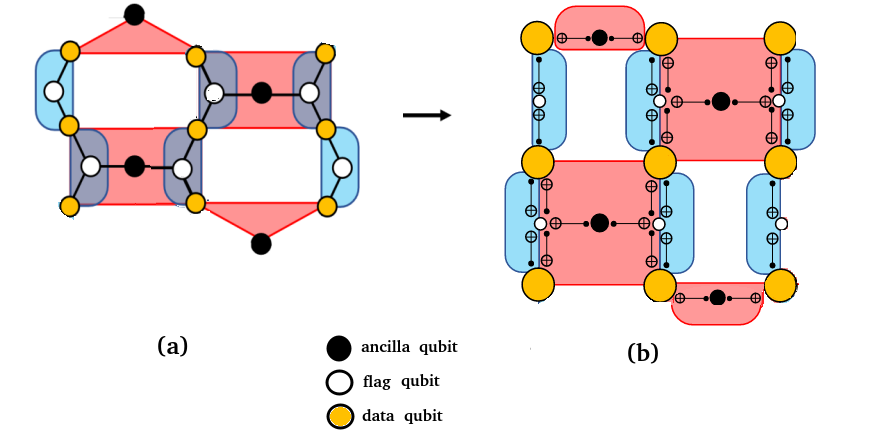}
    \caption{Distance 3 heavy hexagonal code encoding one logical qubit: (a) the hexagonal structure, (b) the circuit illustration of the heavy hexagonal code with the CNOT gates.
    Here yellow, white and black circles represents  data,
    flag and ancilla qubits respectively;  black ancilla qubits are for measuring the $X$ (red face or plaquette) and $Z$ (blue face or strip)  gauge generators. The product of two $Z$  gauge generators at each white plaquette forms a $Z$ stabilizer \cite{chamberland2020topological}. }
    \label{fig:d3hex}
\end{figure}

 \begin{figure}[htb]
    \centering
    \includegraphics[scale = 0.5] {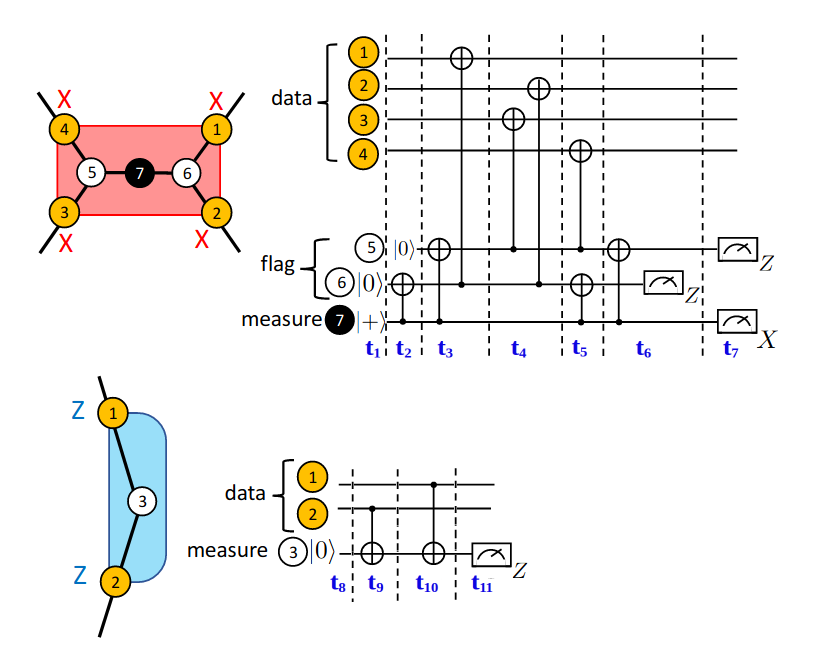}
    \caption{Circuits for measuring $X$ and $Z$ gauge generators in the heavy hexagonal code where $t_i$ denotes the $i^{th}$ time step. Two flag qubits (white circles) measure a $X$  gauge
generator having a weight of 4 and one flag qubit measures a $Z$ gauge
generator having a weight of 2 \cite{chamberland2020topological}.  }
    \label{fig:ckt}
\end{figure}

\subsection{Heavy Hexagonal Code}
This QECC encodes a logical qubit over 
a hexagonal lattice. As qubits are present on both the vertices and edges of the lattice,  the term heavy is used. This is a combination of degree-2 and degree-3 qubits hence there is a huge improvement in terms of average qubit degree in comparison with surface code structure which has qubits of degree-4 \cite{chamberland2020topological}. Fig. \ref{fig:d3hex} shows the lattice for a distance-3 heavy hexagonal code encoding one logical qubit. 

The heavy hexagonal code is a combination of surface code and subsystem code (Bacon Shor code) \cite{chamberland2020topological}.  A subsystem code is defined by $G$, a set of gauge operators where $\forall$ $g \in G$, $\ket{\psi} \equiv g\ket{\psi}$ \cite{bacon2006operator}. A gauge operator takes a codeword to an equivalent subsystem. In other words, a codespace in a subsystem code consists of multiple equivalent subsystems.  It is to be noted that the gauge operators are not necessarily commutative.  The product of two or more gauge operators forms a stabilizer, which keeps the codeword unchanged.

\subsubsection{Gauge generators}
For the heavy hexagonal code, its gauge generators which form the gauge group  are defined in terms of Pauli operators as 
$$ \langle Z_{i,j} Z_{i+1,j},~X_{i,j}X_{i,j+1}X_{i+1,j} X_{i+1,j+1},~
X_{1,2m-1} X_{1,2m},~ X_{d,2m} X_{d,2m+1} \rangle,$$
where $i \in\{1,2,\ldots , d-1\}$, $j \in\{1,2,\ldots , d\}$ and $m \in \{1,2,\ldots , (d-1)/2)\}$. Further, for the second type of the gauge generators $X_{i,j}X_{i,j+1}X_{i+1,j} X_{i+1,j+1}$, $(i+j)$ has to be odd.
A gauge generator $g_{i,j}$ is a gauge operator acting on the $(i,j)^{th}$ data qubit in the lattice. The product of one or more gauge generators can form a \emph{gauge operator}. 

Fig.~\ref{fig:ckt} shows  the $X$ and $Z$  gauge generators   along with the circuits \cite{chamberland2020topological} for measuring these. A single error correction cycle requires 11 time-steps (7 for $X$ and 4 for $Z$). The weight of a gauge operator $g$ is defined as the number of non-identity Pauli operators in  $g$.
In Fig.~\ref{fig:d5hexgauge}, the $Z$ and $X$ gauge generators, indicated in blue and red, can correct bit flip and phase flip errors  respectively. This is explained with detailed examples in Appendix.

For a distance $d$, the number of gauge generators  is $(d^2-1)/2$, so the number of possible gauge  operators can be exponential in $d$, thereby syndrome decoding poses a computational challenge.\\


 \begin{figure} [ht]
    \centering
    \includegraphics[scale = 0.4] {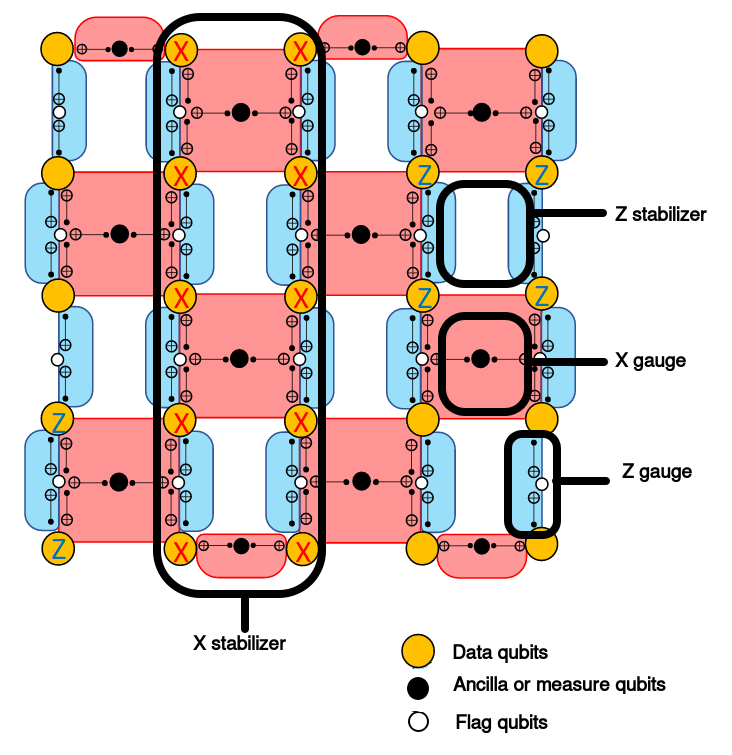}
    \caption{The stabilizers and gauge generators for a distance 5 heavy hexagonal code: the weight-4 $X$ gauge generators are in the red plaquettes, while the weight-2 $X$ gauge generators are on the upper and lower boundaries. Each blue strip denotes  weight-2 $Z$ gauge generators. A vertical strip of two adjacent columns with $X$  gauge generators form an $X$ stabilizer. The weight-4 $Z$ gauge generators in the white plaquettes, and weight-2 $Z$  gauge generators on the left and right boundaries are themselves $Z$ stabilizers. \cite{chamberland2020topological}. }
    \label{fig:d5hexgauge}
\end{figure}

\subsubsection{Stabilizers}

The stabilizer group of the heavy hexagonal code \cite{chamberland2020topological} is defined as 
$$\langle Z_{i,j} Z_{i,j+1} Z_{i+1,j} Z_{i+1,j+1}, ~Z_{2m-1,d} Z_{2m,d}, ~Z_{2m,1} Z_{2m+1,1}, ~\Pi_i X_{i,j} X_{i,j+1} \rangle$$ where $i \in\{1,2,\ldots , d-1\}$, $j \in\{1,2,\ldots , d\}$ and $m \in \{1,2,\ldots , (d-1)/2)\}$. Further, for the first type of stabilizers $Z_{i,j} Z_{i,j+1} Z_{i+1,j} Z_{i+1,j+1}$ having weight 4,  $(i+j)$ has to be even.  The measurement outcome of one such stabilizer is the product of the measured eigenvalues of the two weight-two gauge generators $Z_{i,j} Z_{i+1,j}$ and $Z_{i,j+1} Z_{i+1,j+1}$,  $i + j$ being even.



This is explained with detailed examples in the Appendix. For QECCs which are not subsystem codes, such as \cite{Shor:1997:PAP:264393.264406, PhysRevLett.77.793},   an $n$-qubit QECC with $n-k$ stabilizers can correct upto $k$ errors. On the other hand, subsystem codes need fewer stabilizers \cite{bacon2006operator}. Errors and the corresponding syndromes do not form a one-to-one mapping. 
This has motivated us to design an ML based syndrome decoder, along with two efficient algorithms to identify equivalence classes for all possible errors (details given in Section 4) to enhance the efficiency of the decoder.

\subsection{Noise Models}

The physical errors in the qubits can be of different types. In this subsection we discuss briefly the noise models that are used for this work.

\subsubsection{Bit flip error}
The action of a bit flip error \cite{nielsen2002quantum} on a quantum state $\rho$ is denoted as
$$\rho \rightarrow (1-p_x)\rho + p_x X \rho X^{\dagger},$$
$p_x$ being the probability that an unwanted Pauli $X$ error occurs.
\subsubsection{Phase flip error}

The evolution of the state in a phase flip \cite{nielsen2002quantum} is given as
$$\rho \rightarrow (1-p_z)\rho + p_z Z \rho Z^{\dagger}$$
where $p_z$ is the probability that an unwanted Pauli $Z$ error occurs.

\subsubsection{Depolarization noise}
The evolution of a quantum state $\rho$ under  depolarization noise \cite{nielsen2002quantum} is given as

$$\rho \rightarrow (1-p)\rho + \frac{p}{3} X \rho X^{\dagger} + \frac{p}{3} Y \rho Y^{\dagger} + \frac{p}{3} Z \rho Z^{\dagger}$$

where $p$ is the probability of error. It is basically a depolarizing channel with error rate $p$ \cite{li20192d}.

\subsubsection{Measurement error} After measuring the ancilla, we get the syndrome. Depending on the measurement error probability, bit flip error is incorporated in that syndrome.

\subsubsection{ Stabilizer error} Stabilizer error is basically erroneous measure qubits. Hence bit flip errors are applied on ancilla qubits. 

Moreover, a single error correction cycle in heavy hexagon code consists of eleven steps (Fig. \ref{fig:ckt}. An error can occur on $0 \leq k \leq d^2$ data qubits in each of the eleven steps, $d$ being the distance of the code. Therefore, the overall probability of error for depolarizing error for each error correction cycle is $1 - (1-p)^{11}$.

\subsection{Machine Learning based decoding}
Minimum Weight Perfect Matching algorithm (MWPM)  has been used for decoding the heavy hexagonal code \cite{chamberland2020topological}. Machine learning based decoders have been shown to outperform MWPM for surface codes \cite{baireuther2018machine,varsamopoulos2017decoding,krastanov2017deep,varsamopoulos2019comparing, bhoumik2022ml}. In this article we have designed ML based decoders for the heavy hexagonal code, and have shown that the subsystem property of this code leads to even more efficient decoding than surface code. We elaborate on those properties in Sec.~\ref{sectiongauge}. 

\subsubsection{Motivation for using Machine Learning based syndrome decoder}
Topological codes such as surface code, heavy hexagonal code are degenerate because there exists errors $e_1 \neq e_2$ such that $e_1\ket{\psi} = e_2\ket{\psi}$, $\ket{\psi}$ being the codeword \cite{fowler2012towards}. This results in any decoder failing to distinguish between $e_1$ and $e_2$. But this failure does not always result in a logical error \cite{bhoumik2022ml}. A logical error \textit{usually}  occurs when a decoder fails to distinguish between $\lfloor \frac{d-1}{2} \rfloor$ and $\lceil \frac{d+1}{2} \rceil$ errors \cite{fowler2012towards}.

Classical algorithms such as Minimum Weight Perfect Matching (MWPM) \cite{edmonds_1965}, used for decoding topological error correcting codes \cite{fowler2012towards, chamberland2020topological} aim to determine the minimum number of errors which can recreate the obtained error syndrome without considering the probability of various errors. Therefore, such a decoder is largely vulnerable to inducing a logical error due to incorrect decoding. 

Machine learning techniques can be applied to overcome this drawback and also consider the probability of error in the system while decoding such that it can predict the best possible error correction accordingly \cite{baireuther2018machine,varsamopoulos2017decoding,krastanov2017deep,varsamopoulos2019comparing, bhoumik2022ml} with comparatively lower time complexity \cite{varsamopoulos2019comparing}. It learns the probability of error for each type of error and predicts the most likely one among  the $\lfloor \frac{d-1}{2} \rfloor$ and $\lceil \frac{d+1}{2} \rceil$  type of  errors.

\section{Designing ML based decoder for heavy hexagonal code}

Artificial neural networks (ANN) are brain-inspired techniques for replicating the procedure of how we humans learn and they are heavily used in machine learning. Neural networks consists of a single input and output layer along with a few hidden layers. The hidden layers transform the input into an intermediate form and the output layer finds patterns from its previous hidden layers. The time complexity to train a neural network with $m$ input nodes, one hidden layer with $h$ nodes and $L$ output nodes is $\mathcal{O}((m+L) \cdot h)$. In this paper we are using feed forward neural network for decoding heavy hexagonal code of distance 3, 5, and 7 .

For application of ML in decoding, we first reduce the decoding problem to classification, a well-studied problem in machine learning. Classification is the process of predicting the class of given data points. Classes are also known as labels. It is the task of approximating a mapping function $f$ from input variables $x$ to output variables $y$. The methodology to map heavy hexagonal code is similar with that for surface code which has already been discussed in detail in \cite{bhoumik2022ml}.

A distance $d$ heavy hexagonal code has  (i) $(d^2 -1)/2$ syndrome bits in case of bit flip error (refer to $Z$ stabilizers in Fig. \ref{fig:d5hexgauge});    (ii) $d-1$  syndrome bits in case of phase flip error (refer to  $X$ stabilizers in Fig. \ref{fig:d5hexgauge}). This syndrome is the input data to the ML model. The label of the ML model is the erroneous data qubit of the heavy hexagonal code structure, hence it has $d^2$ qubits in each entry.  Our syndrome decoding problem is mapped into a \emph{multi-class multi-label  classification} problem in which there are ${2^{d^2}}$ classes, and each class label consists of $d^2$ bits. In our feed forward neural network, the input layer gets the syndrome (measured ancilla qubits) and the output layer identifies the type of error and its location in the lattice.

\section{Methods to reduce error classes for heavy hexagonal code}

\label{sectiongauge}


The subsystem property of a QECC asserts that if $\Pi_j g_j$ denotes the product of one or more gauge operators, then
\begin{itemize}
    \item  an error $e = \Pi_j g_j$ can be safely ignored as the system transforms it to an equivalent subsystem;
    \item if for two errors $e_1$ and $e_2$, $e_1 \Pi_j g_j = e_2$, then $e_2$ can be considered as $e_1$ as both take the state to equivalent erroneous subsystems.
\end{itemize}

We term the second scenario  as \emph{gauge equivalence}, which provides significant advantage in designing ML based decoders. For both $\sigma_x$ (bit flip) and $\sigma_z$ (phase flip) errors, a distance $d$ heavy hexagonal code having $d^2$ qubits mandates a classification of the $2^{d^2}$ possible errors, each being  termed as an \textit{error class} henceforth. However, for a subsystem code such as the heavy hexagonal code, there exists $i \neq j$ such that $Q_i$ is gauge equivalent to $Q_j$. From now on, if $\ket{\psi} \equiv \Pi_j g_j \ket{\phi}$, where $\Pi_j g_j$ denotes the product of one or more gauge operators $g_j$, then we write $\ket{\psi} \equiv \ket{\phi}$ modulo $(\Pi_j g_j)$.

For example, in Fig \ref{fig:xgauge} there are 4 $X$-gauge generators $G1, G2, G3, G4$, 6 $Z$-gauge generators $g1, g2, g3, g4, g5, g6$ and the qubits are $Q1, Q2, ..., Q9$. If $\ket{\psi}$ is the codeword, then, 
\begin{itemize}
    \item $X_4 X_7 X_8 \ket{\psi} \equiv X_5 \ket{\psi}$ modulo ($G_3$), where $X_k$ denotes bit flip error on qubit $Q_k$.
    
    \item $Z_7 \ket{\psi} \equiv Z_1 \ket{\psi}$ modulo ($g_4 g_1$), where $Z_k$ denotes phase flip error on qubit $Q_k$.
   
\end{itemize}

 \begin{figure}
    \centering
    \includegraphics[height=10 cm, width= 9 cm] {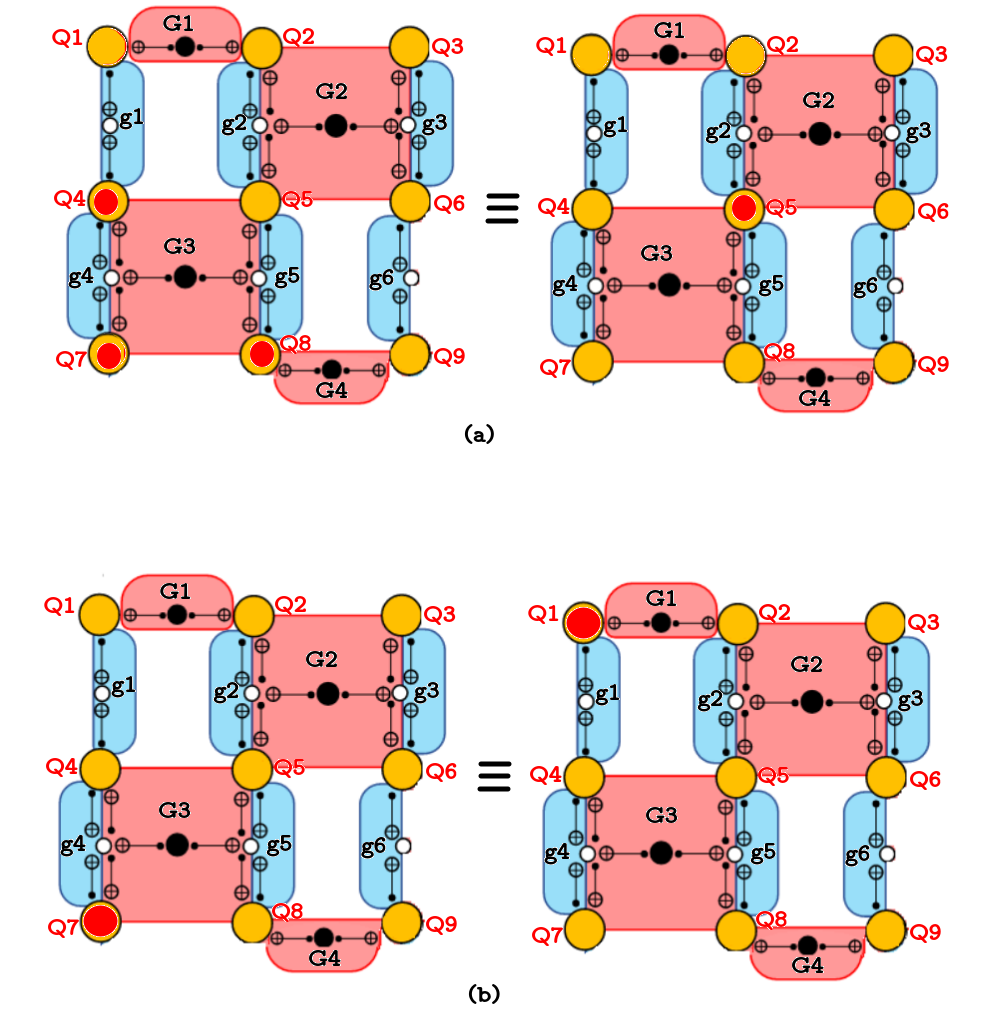}
    \caption {(a) $X$ gauge equivalence for bit flip error: simultaneous errors on data qubits 4, 7 and 8  is equivalent to an error on data qubit 5, because by applying $X$ gauge operator $G3$ (consisting of $X$ operators on data qubits 4, 5, 7 and 8) on data qubits 4, 7 and 8, data qubit 5 has an error.  (b) $Z$ gauge equivalence for phase flip errors: an error on data qubit 7 is equivalent to an error on data qubit 1 because by applying $Z$ Gauge operator $g4$ (consisting of $Z$ operators on qubits 4 and 7) followed by $Z$ Gauge operator $g1$ (consisting of $Z$ operators on qubits 1 and 4)  on data qubit 7, data qubit 1 has an error.}
    \label{fig:xgauge}
\end{figure}

The notion of gauge equivalent error strings creates a problem for machine learning based decoding since the problem of mapping syndromes to error strings is not well defined because there can be multiple Pauli error strings which are gauge equivalent  having the same syndrome. In order to remedy the above, we identify a representative element from each such an error class. Every error string in the training data is mapped to the representative element of its corresponding class.  
Next, can we reduce the number of error classes? Reduction in the number of error classes implies a classification problem with fewer classes, which helps in improving the performance further for the ML model. 
We now formally define the criteria for two qubits to belong to the same error class.

\begin{lemma}
\label{equiv}
Given a codeword $\ket{\psi}$ such that $\ket{\psi} \equiv \Pi_{j} g_j \ket{\psi}$, where $\Pi_{j} g_j$ implies the product of one or more gauge operators $g_j$, any error $e$ acting on the codeword is equivalent to $e (\Pi_{j} g_j)$.
\end{lemma}

Proof : 
Consider an error $e$ acting on the codeword $\ket{\psi}$ such that $\ket{\psi} \equiv \Pi_{j} g_j \ket{\psi}$. Therefore,
\begin{eqnarray}
e\ket{\psi} &\equiv& e \Pi_{j} g_j \ket{\psi} \Rightarrow e \equiv e (\Pi_{j} g_j) \nonumber
\end{eqnarray}

Lemma~\ref{equiv} implies that $e$ and $e(\Pi_{j} g_j)$ as the same error class. Such equivalence leads to a set of error classes whose cardinality is less than $2^{d^2}$, where errors in the same class act equivalently on the codeword. Therefore, it is not necessary to distinguish between the errors in the same \textit{error class}. It suffices to identify the error class only.

 We propose two methods to find gauge equivalence and determine the error classes for the heavy hexagonal code, namely \emph{search based equivalence} and \emph{rank based equivalence}.

\subsection{Reducing bit flip error classes by search based gauge equivalence}

The heavy hexagonal code, being similar to surface code \cite{bravyi1998quantum} and Bacon-Shor code \cite{bacon2006operator} for bit flip and phase flip errors respectively \cite{chamberland2020topological}, we 
employ the subsystem code property  to modify our training dataset into another equivalent training dataset with fewer class labels.
In this subsection we present an algorithm to convert a given trining data set to an equivalent one  having fewer class labels for bit flip errors.



Before presenting the Algorithm, we prove in Lemma~\ref{lemma1} the number of error classes obtained due to the subsystem property. As multiple errors can fall under the same class, we require a class representative. We show a method to find the class representative, followed finally by the Algorithm to obtain such equivalence.

\begin{lemma}
\label{lemma1}
For a distance $d$ heavy hexagonal code, the total number of bit flip error classes  is $2^{\frac{d^2+1}{2}}$.
\end{lemma}

Proof :
For a distance $d$ heavy hexagonal code, the total number of data qubits is $d^2$ and the number of $X$ gauge generators is $\frac{d^2-1}{2}$. Therefore, the total possible combinations of these can provide $2^{\frac{d^2-1}{2}}$ $X$ gauge operators. For each of these operators $g$, any error $e \equiv e.g$. Therefore, $e$ and $e.g$ belong to the same error class. The total number of error classes is
$2^{d^2}/2^{\frac{d^2-1}{2}} = 2^{\frac{d^2+1}{2}}$.

We thus obtain a  square root order reduction  in the number of error classes.  Table~\ref{table:1}  presents the reduction in the number of error classes for codes of distance 3, 5 and 7 due to gauge equivalence. As multiple errors belong to the same error class, a \textit{class representative} needs to be chosen efficiently for our ML decoder.

\begin{table}[htb]
\centering
\caption{Number of bit flip error class labels without and with gauge equivalence }
\begin{tabular}{ |c|c|c| }
 \hline
Code distance& \# Class labels without equivalence& \# Class labels with equivalence\\
 \hline
 3 & $2^9$ & $2^5$ \\
 5& $2^{25}$ & $2^{13}$ \\
 7& $2^{49}$ & $2^{25}$\\
 \hline
\end{tabular}

\label{table:1}
\end{table}
 
\textit{Choice of class representative}:  If errors $e_1, e_2, \hdots, e_k$ belong to the same error class $\mathcal{E}$, then we choose a class representative $e_i \in \mathcal{E}$ which is to be considered as the error on the codeword for any error $e \in \mathcal{E}$ occurring on the codeword. For example, in Fig.~\ref{fig:xgauge} it can be verified that bit flip errors on qubits 4, 7 and 8 produces the same syndrome as that on 5 . But for training the ML decoder, we choose the second one as the common label for both as it has a lower weight. However, simply considering the weight of the errors is not sufficient to obtain a class representative. For example, bit flip errors on qubits 1 and 2 lead to the same syndrome, and both of these errors have weight 1. Therefore, we use the method of \emph{lexicographic minima} to obtain the class representative.


For finding the lexicographic minima, we represent the qubits in the codeword as a characteristic vector of size $d^2$, where the $i^{th}$ bit is 1 if an error occurred on qubit $i$, and 0 otherwise. For example, the  bit flip errors on qubits 4, 7 and 8 in Fig. ~\ref{fig:xgauge} is denoted as  $e = [0 0 0 1 0 0 1 1 0]$. For each error class $\mathcal{E}$, we choose error $e$ as the class representative for which $L(e) = \displaystyle \sum_{i=0}^{d^2-1} 2^i*e[i]$ is minimum. For example, in a distance 3 heavy hexagonal code, bit flip error on qubits 1 and 2  are written as $e_1 = [1 0 0 0 0 0 0 0 0]$ and $e_2 = [0 1 0 0 0 0 0 0 0]$ respectively. We note that $X$-operator on qubits 1 and 2 form an $X$ gauge operator,  which when applied on $e_1$ gives us $e_2$, and vice-versa. So $e_1$ and $e_2$ belong to the same error equivalence class $\mathcal{E}$. Here, $L(e_1) = 2^1*1 + 2^2*0 + 2^3*0 + ... + 2^9*0  = 2$ and $ L(e_2) = 4$, so $e_1$ is chosen as the class representative.

 Algorithm~\ref{alg:gauge}  finds the gauge equivalent errors for each error $e$ in the training dataset  $E[1:N]$. Its  time and space complexity are given in Lemmata~\ref{bitcomplexity} and ~\ref{bitspace}.

\begin{algorithm}[H]
\caption{Search based method to find gauge equivalence class representatives for bit flip errors}
\label{alg:gauge}
\begin{algorithmic}[1]
\REQUIRE  $G_x$, the list of $X$-gauge operators; $d$, distance of the heavy hexagonal code; $E[1:N]$, the list of all errors $e_1, e_2, ..., e_N$ in the training dataset.

\ENSURE for all errors in $E[1:N]$, equivalent error vector $MinWeightEquiv$ having weight $MinWeight$
\FOR{$e = e_1$ to $e_N$}
\STATE $MinWeightEquiv= e$
\STATE  $MinWeight = \displaystyle \sum_{i=0}^{d^2-1} 2^i*e[i]$

\FORALL { $g_x \in G_x$}
\STATE  $GaugeEquiv=g_x\oplus e$ \* applying $g_x$ on $e$*\
\STATE  $weight_{GE}$ = $\displaystyle \sum_{i=0}^{d^2-1} 2^i*GaugeEquiv[i]$
\IF {$MinWeight > weight_{GE}$}
\STATE $MinWeight =  weight_{GE}$ 
\STATE $MinWeightEquiv = GaugeEquiv $
\ENDIF
\ENDFOR
\STATE return $MinWeightEquiv$
\ENDFOR
\end{algorithmic}
\end{algorithm}


\begin{theorem}
\label{thm:bit}

Given a training sample $e$ in which each qubit may be either error-free or has a bit flip error only, its corresponding gauge equivalent class representative $e_{min} \equiv e$ can be computed by Algorithm \ref{alg:gauge} in $\mathcal{O}(N \cdot d^2 \cdot 2^{d^2/2})$ time using $\mathcal{O}(d^2 \cdot 2^{\frac{d^2-1}{2}} + N. d^2)$ space, where $d$ is the distance of the code and $N$ is the number of training samples.
\end{theorem}

Proof:
For an error $e$, let $MinWeightEquiv \equiv e$ be the lexicographic minimum in $E[1:N]$ which is gauge equivalent to $e$. The gauge operators and errors are represented as $\{0,1\}$ strings. Let $L(e)$ denote the lexicographic value corresponding to the binary characteristic vectors representation of $e$. Therefore, $L(MinWeightEquiv) \leq L(e_{equiv})$ $\forall$ $e_{equiv} \equiv e$. Algorithm~\ref{alg:gauge} finds $L(e_{equiv})$ in decimal for all $e_{equiv} \equiv e$, and returns the vector  $MinWeightEquiv$  which has the minimum value. 
Hence, $L(MinWeightEquiv) \leq L(e_j)$ $\forall$ $e_j \equiv e$. Therefore, $MinWeightEquiv = e_{min}$.

The proofs for time and space complexity are given below in Lemmata~\ref{bitcomplexity} and ~\ref{bitspace} respectively.

\begin{lemma}
\label{bitcomplexity}
Given a training sample $e$ in which each qubit may be either error-free or have bit flip error only, Algorithm \ref{alg:gauge} computes its corresponding gauge equivalent class representative $e_{min} \equiv e$  in $\mathcal{O}(N \cdot d^2 \cdot 2^{d^2/2})$ time, where $d$ is the distance of the code and $N$ is the number of training samples. 
\end{lemma}
Proof:  For a distance $d$ heavy hexagonal code, the total number of data qubits is $d^2$ and the    number of $X$ gauge generators is $\frac{d^2-1}{2}$. Therefore, the total possible combinations of the $X$ gauge operators is $2^{\frac{d^2-1}{2}}$.Each of the line numbers 2, 5, and 6 of Algorithm \ref{alg:gauge} requires $d^2$ bitwise operation over the length of an error string.

There are $2^{\frac{d^2-1}{2}}$ possible gauge operators. Therefore the inner loop (line 4 to 11) runs $2^{\frac{d^2-1}{2}}$ times. So the time for each training sample is $d^2 +  d^2 \cdot 2^{\frac{d^2-1}{2}} + c$, where $c$ is a constant, leading to a time complexity of $\mathcal{O}(d^2 \cdot 2^{d^2/2})$ for the inner \emph{for} loop of Algorithm \ref{alg:gauge}. Therefore, the total time complexity of Algorithm \ref{alg:gauge} to find the gauge equivalence for the entire training dataset consisting of $N$ samples is $\mathcal{O}(N \cdot d^2 \cdot 2^{d^2})$.


\begin{lemma}
\label{bitspace}
Given a training sample $e$ in which each qubit may be either error-free or has a bit flip error only, Algorithm \ref{alg:gauge} computes its corresponding gauge equivalent class representative $e_{min} \equiv e$  with space complexity of $\mathcal{O}(d^2 \cdot 2^{\frac{d^2-1}{2}} + N. d^2)$, where $d$ is the distance of the code and $N$ is the number of training samples.
\end{lemma}

Proof:
For a distance $d$ heavy hexagonal code, there are $d^2$ data qubits in the lattice for a logical qubit and $\frac{d^2-1}{2}$ $X$ gauge generators. 

    Each of the $2^{\frac{d^2-1}{2}}$ possible $X$ gauge operators need $d^2$ bits. Hence the total space required to store all the gauge operators is $\mathcal{O}(d^2 \cdot 2^{\frac{d^2-1}{2}})$. Furthermore, the entire training dataset of $N$ samples needs space complexity of  Algorithm \ref{alg:gauge} is  $\mathcal{O}(d^2 \cdot 2^{\frac{d^2-1}{2}} + N. d^2)$.

\begin{table}[h]
\centering
\caption{Number of error class labels without and with gauge equivalence in case of phase flip }
\begin{tabular}{ |c|c|c| }
 \hline
Code distance& \# Class labels without equivalence& \# Class labels with equivalence\\
 \hline
 3 & $2^9$ & $2^3$ \\
 5& $2^{25}$ & $2^{5}$ \\
 7& $2^{49}$ & $2^{7}$\\
 \hline
\end{tabular}

\label{table:2}
\end{table}

In the expression of space complexity  $\mathcal{O}(d^2 \cdot 2^{\frac{d^2-1}{2}} + N. d^2)$ (Lemma 7), the second term ($N. d^2$) dominates when $N > 2^{\frac{d^2-1}{2}} $.

Instead of calculating all possible combinations of $X$ gauge operators beforehand, if we calculate them on the go, then the space complexity can be reduced to $\mathcal{O}(d^4 + N. d^2)$. However, for such an algorithm, the entire set of gauge operators need to be generated for every training sample, leading to a time complexity of upto $\mathcal{O}(N \cdot d^4 \cdot 2^{d^2/2})$.

\subsection{Reducing error classes for phase flip using search based gauge equivalence}

In this subsection, we present studies for phase flip error similar to those for bit flip error in the previous subsection. The Bacon-Shor structure \cite{bacon2006operator} of phase stabilizers allow for more equivalent subsystems, leading to fewer error classes. We first depict this equivalence for phase flip errors in Lemma~\ref{phase}, and provide the algorithm to find the class representatives for gauge equivalent phase flip error as Algorithm~\ref{alg:phasegauge}.

\begin{lemma}
\label{phase}
For any column $1 \leq j \leq d$ of physical qubits in the lattice for heavy hexagonal QECC, the phase operator $Z_{i,j}$ on $i^{th}$ qubit of the $j^{th}$ column is gauge equivalent to the phase operator $Z_{1,j}$.
\end{lemma}

Proof:
(by induction)\\
\textit{Base Case}: If a phase flip error occurs on the first qubit of any column of qubits $j$, then it is trivially equivalent to $Z_{1,j}$.

\textit{Induction Hypothesis}: For all positive integer $k$, $1 \leq k < d$, $Z_{k,j}$ is gauge equivalent to $Z_{1,j}$.

\textit{Induction Step}: Consider a  phase operator $Z_{k+1,j}$. Now, for all $k$, $1 \leq k < d$, there exists a $Z$ gauge operator given by
$Z_{k,j}\cdot Z_{k+1,j}$. Therefore, $Z_{k+1,j}$ is gauge equivalent to $(Z_{k+1,j}) \cdot (Z_{k,j} Z_{k+1,j}) = Z_{k,j}$. By the induction hypothesis, as $Z_{k,j}$ is gauge equivalent to $Z_{1,j}$,  $Z_{k+1,j}$ is  gauge equivalent to $Z_{1,j}$.

\begin{corollary}
There are $2^d$ non-equivalent error classes for phase flip errors.
\end{corollary}

Proof:
This follows from Lemma~\ref{phase} because $q_j^{i}$ is gauge equivalent to $q_j^{1}$ for the $j^{th}$qubit column, $1 \leq i \leq d$,  each qubit column corresponds to one possible phase flip error class. There being $d$ qubit columns, there are $2^d$ non-equivalent phase flip error classes.

Algorithm~\ref{alg:phasegauge} below finds the gauge  equivalence class representatives for each error $e$ in the training dataset.  Its time and space complexity analyses appear in Lemmata~\ref{phasecomplexity} and ~\ref{phasespace}.

\begin{algorithm}[H]
\caption{Search based method to find gauge equivalence class representatives for phase flip errors}
\label{alg:phasegauge}
\begin{algorithmic}[1]
\REQUIRE $d$, distance of the code, error string  $E[1:N]$, the list of all errors $e_1, e_2, ..., e_N$ in the training dataset.

\ENSURE for all errors in $E[1:N]$, equivalent error vector $MinWeightEquiv$ having weight $MinWeight$
\FOR{$e = e_1$ to $e_N$}

\FOR{$j=d$ to $d^2 -1$}
\STATE $l=j\%d$
\STATE $e[l] = e[l] \oplus e[j]$
\ENDFOR
\STATE $MinWeightEquiv$ = $e$
\STATE return $MinWeightEquiv$
\ENDFOR
\end{algorithmic}
\end{algorithm}

\begin{theorem}
\label{thm:phase}
Given a training sample $e$ in which each qubit may be either error free or having phase flip error only, its corresponding gauge equivalent minima $e_{min} \equiv e$ can be computed by Algorithm \ref{alg:phasegauge} in $\mathcal{O}(N.d^2)$ time using  $\mathcal{O}(N.d^2)$ space, where $d$ is the distance of the code and $N$ is the number of training samples.
\end{theorem}

Proof:
For an error $e$, let $e_{min}$ be the lexicographic minimum corresponding to $e$. According to Lemma 8, for any column $1 \leq j \leq d$ of qubits where $Z_{i,j}$ denotes the phase operator on $i^{th}$ qubit of the $j^{th}$ column (where $1 \leq i \leq d$),  $Z_{i,j}$ is gauge equivalent to $Z_{1,j}$. For an error $e = \{e^1 e^2 \hdots e^k\}$, where $e^i \in \{I, Z\}$, $1 \leq i \leq k$ ($k$ is the number of qubits), Algorithm~\ref{alg:phasegauge} finds the gauge equivalence $e^i _{equiv}$ of each $e^i$ using Lemma~\ref{phase}. The final gauge equivalent error $e_{equiv}$ is determined as 
$e_{equiv} = \{e^1 _{equiv} \cdot e^2 _ {equiv} \hdots e^k _{equiv}\}$. As it finds the equivalent error as the top most row's error in the lattice, hence it will be of the minimum weight. Hence $e_{equiv} = e_{min}$
The proofs for time and space complexity are given below in Lemmata~\ref{phasecomplexity} and ~\ref{phasespace} respectively.

\begin{lemma}
\label{phasecomplexity}
Given a phase flip error $e$, its corresponding gauge equivalent minima $e_{min} \equiv e$ can be computed by Algorithm \ref{alg:phasegauge} in $\mathcal{O}(N. d^2)$ time, where $d$ is the distance of the code.
\end{lemma}

Proof:
 Lemma \ref{phase} asserts that for each column of the heavy hexagonal code lattice, a phase flip error on a qubit in that column is gauge equivalent to the phase flip error on the top most qubit of that column. Therefore, Algorithm \ref{alg:phasegauge} starts from the second row of the lattice, leaving aside the topmost row (length $d$) which consists of the first qubits of each column. Each of line numbers 3 and 4 of Algorithm \ref{alg:phasegauge} requires a constant time operation for  $d^2-d$ error strings. The outer loop, i.e. line number 1 runs $N$ times, as the size of the entire training dataset is $N$. Therefore, the time complexity of Algorithm \ref{alg:phasegauge} is $N \cdot c \cdot (d^2-d)$, where $c$ is a constant, i.e., $\mathcal{O}(N. d^2)$.

It is to be noted that we obtain a quadratic reduction in the number of error classes for both bit flip and phase flip errors by employing gauge equivalence but the reduction is much higher for phase flip errors.

\begin{lemma}
\label{phasespace}
Given a phase flip error $e$, its corresponding gauge equivalent minima $e_{min} \equiv e$, computed by Algorithm \ref{alg:phasegauge} has a space complexity of $\mathcal{O}(N.d^2)$ where $d$ is the distance of the code, and $N$ number of training samples.
\end{lemma}

Proof:
 In order to store $N$ number of training samples we need $\mathcal{O}(N. d^2)$ amount of space.  Therefore, the total space needed for Algorithm \ref{alg:phasegauge} is $\mathcal{O}(N. d^2)$ .

We observe that based on linear search the time complexity for finding equivalence class representatives of bit flip errors is $\mathcal{O}(N.d^2 \cdot 2^{d^2})$ and for phase flip errors is $\mathcal{O}(N.d^2)$. In the next subsection, we  propose a method called rank based gauge equivalence which can further reduce the time complexity for the case of bit flip errors, which is much higher than that for phase flip errors.

\subsection{Reducing error classes for bit flip using rank based gauge equivalence }

This algorithm relies on finding the rank of a matrix $M = [GEN_x]$, constructed with a column corresponding to each of  the $X$ gauge generators. Recall that the rank of a matrix denotes the number of linearly independent columns in it \cite{strang2006linear}. Since the gauge generators are linearly independent, $M$  exhibits full rank equal to $(d^2-1)/2$, the number of $X$ gauge generators. 

For two gauge equivalent errors $e_i$ and $e_k$,  if we obtain another matrix $M' = [GEN_x, e_i - e_k]$  by appending one more column to $M$, then $M$ and $M'$ has the same rank because $e_k$ can be expressed as $(\Pi_j g_j) e_i$. We use this notion to find the gauge equivalence for bit flip errors faster than Algorithm~\ref{alg:gauge}.


In Algorithm~\ref{alg:rankgauge}, we provide an algorithm to find the equivalent for each error $e$ in the training dataset, and depict its time and space complexities in Lemmata~\ref{rank} and ~\ref{rankspace} respectively.

The search space of Algorithm~\ref{alg:rankgauge} differs from that of Algorithm~\ref{alg:gauge}. Let $EQUIV(e)$ denote the set of all gauge equivalent errors corresponding to an error $e$. The cardinality of $EQUIV(e)$ scales exponentially in the number of gauge \emph{generators}. On the other hand, Algorithm~\ref{alg:rankgauge} determines $EQUIV_{TS}(e)$ which is the set of all gauge equivalent errors corresponding to an error $e$ present in the set of training samples. Therefore, $EQUIV_{TS}(e) \subseteq EQUIV(e)$. Since Algorithm~\ref{alg:rankgauge} only searches for the gauge equivalent errors present in the training samples, the search space for this algorithm is never greater than that for Algorithm~\ref{alg:gauge}. Both  Algorithms~\ref{alg:gauge} and ~\ref{alg:rankgauge} reduce the error classes for bit flip using gauge equivalence. While in Algorithm~\ref{alg:gauge} the class label for the equivalence classes is of minimum possible weight, in Algorithm~\ref{alg:rankgauge} the class label is the minimum possible weight for errors present in the training set. But our target is to reduce the number of error classes. So both the algorithms serve that purpose. It is to be noted that the equivalence error classes  remain the same for both methods. Only the class representatives are likely to differ.

\begin{theorem}
\label{thm:rank}
Given a training sample $e$ in which each qubit may be error free or has bit flip error only, its corresponding gauge equivalent minima $e_{min} \equiv e$ can be computed, by Algorithm \ref{alg:rankgauge} in $\mathcal{O}(N^2 \cdot d^4)$ time using  $\mathcal{O}(d^4+N.d^2)$ space, where $d$ is the distance of the code and $N$ is the number of training samples.
\end{theorem}

Proof:
We prove this theorem in two steps. The proofs for time and space complexity are shown individually in Lemmata~\ref{rank} and ~\ref{rankspace} respectively.

\begin{algorithm}[H]
\caption{Rank based method to find gauge equivalence class representatives for bit flip errors}
\label{alg:rankgauge}
\begin{algorithmic}[1]
\REQUIRE $GEN_x$, the list of $X$ gauge generators, $d$, distance of the code,  List of all errors in the dataset $E[1:j]$ which consists $e_1, e_2, ..., e_j$.

\ENSURE Error string $MinWeightEquiv$ which is equivalent to $e_j$.

\STATE $MinWeightEquiv= {e_j}$
\STATE  $MinWeight = \displaystyle \sum_{k \in d^2}  2^k*{e_j}[k]$

\STATE $M = [GEN_x]$ where each $GEN_x$, $0 \leq x \leq (d^2-1)/2$ forms a column of $M$.

\STATE Calculate the rank $M$

\FORALL{$e = e_1$ to $e_{j-1}$}

\IF {$syndrome(e_j) == syndrome(e)$}

\STATE Form a matrix $M'$ by appending $e_j - e$ as the last column of $M$

\STATE Calculate the rank of $M'$

 \IF {$Rank(M) == Rank(M')$ }

\STATE  $weight_{GE}$ = $\displaystyle \sum_{k \in d^2} 2^k*e[k]$

 \IF {$MinWeight > weight_{GE}$}

 \STATE $MinWeight =  weight_{GE}$
 \STATE $MinWeightEquiv = e $
 \ENDIF

 \ENDIF

 \ENDIF

\ENDFOR

\STATE Return $MinWeightEquiv$
\end{algorithmic}
\end{algorithm}

\begin{lemma}
\label{rank}
Given a bit flip error $e$, its corresponding gauge equivalent minima $e_{min} \equiv e$ can be computed by Algorithm \ref{alg:rankgauge} in $\mathcal{O}(N^2 \cdot d^4)$ where $d$ is the distance of the code and entire dataset consists of $N$ errors.
\end{lemma}

Proof:

In Algorithm \ref{alg:rankgauge}, each of the lines 2 and 10 requires  bitwise operation over the length of the error string
which is $d^2$. Computing the syndrome equality in line 6 requires  bitwise operation over the length of the syndrome string which is $(d^2-1)/2$. There are $d^2$ number of data qubits and the number of gauge generators is $(d^2-1)/2$. Furthermore, calculation of rank of an $m \times n$ matrix is $\mathcal{O}(m \cdot n)$ \cite{cormen2022introduction}.
Therefore, calculating the rank of matrix M of size ($d^2 \times (d^2-1)/2)$) in line 4, has a running time of $\mathcal{O}(d^4)$. For a particular error $e_j$, the time required to calculate the gauge equivalence is given by $\sum_{i=1}^{j-1}\mathcal{O}(d^4 +d^2) $ = $\mathcal{O}((d^4+d^2)\cdot j )$.\\
Hence, for the entire dataset consisting of $N$ errors, the time complexity of Algorithm \ref{alg:rankgauge} is
\begin{eqnarray*}
\mathcal{O}((d^4+d^2)\cdot \sum_{j=1}^{N}j )
&=& \mathcal{O}(N^2 \cdot (d^4 +d^2))\\
&=& \mathcal{O}(N^2 \cdot d^4)
\end{eqnarray*}

\begin{lemma}
\label{rankspace}
Given a bit flip error $e$, Algorithm \ref{alg:rankgauge} computes its corresponding gauge equivalent minima $e_{min} \equiv e$ using  $\mathcal{O}(d^4+N.d^2)$ space where $d$ is the distance of the code and $N$ is the number of training samples.
\end{lemma}

Proof:
For a distance $d$ heavy hexagonal code, the total number of data qubits is $d^2$ and the number of $X$ gauge generators is $\frac{d^2-1}{2}$. Each of the gauge generators consists of $d^2$ bits. Hence, to store the generator matrix $M$, the space needed is $d^2 \cdot (d^2-1)/2$. To store the appended matrix $M'$,  $d^2 \cdot ((d^2-1)/2+1)$ space is needed and  to store $N$ number of training samples, $N \cdot d^2$ space is needed. Hence, the total amount of space needed is $d^2 \cdot (d^2-1)/2 + d^2 \cdot ((d^2-1)/2+1) + N \cdot d^2 $ leading to a
space complexity of $\mathcal{O}(d^4+N.d^2)$.

A pertinent question, therefore, is the criterion for which rank based equivalence is faster than the search based one. This is answered in Corollary~\ref{cor:rankvssearch}.

\begin{corollary}
\label{cor:rankvssearch}
Rank based equivalence method finds gauge equivalence faster than the search based method if $2^{d^2} \geq c_ \cdot (N\cdot d^2) $, where $d$ is the distance of the code, $N$ the number of training samples, and $c$ a constant.
\end{corollary}

Proof:
According to Lemma \ref{bitcomplexity}, given a bit flip error $e$, its corresponding gauge equivalent minima $e_{min} \equiv e$ can be computed, according to Algorithm \ref{alg:gauge}, in $\mathcal{O}(d^2 \cdot 2^{d^2})$, where $d$ is the distance of the code. Hence the total time required to find the gauge equivalence for the entire training dataset consisting of $N$ samples is $\mathcal{O}(N \cdot d^2 \cdot 2^{d^2})$. According to Lemma \ref{rank} given a bit flip error $e$, its corresponding rank based gauge equivalent minima $e_{min} \equiv e$ can be computed by Algorithm \ref{alg:rankgauge} in $\mathcal{O}(N^2 \cdot d^4)$ where $d$ is the distance of the code and entire dataset consists of $N$ errors. Hence, it is gainful to employ rank based gauge equivalence only when
\begin{eqnarray}
    \mathcal{O}(N \cdot d^2 \cdot 2^{d^2}) &\geq& \mathcal{O}(N^2 \cdot d^4) \nonumber \\
    \Rightarrow c_1 \cdot N \cdot d^2 \cdot 2^{d^2} &\geq& c_2 N^2 \cdot d^4 \nonumber \\
    \Rightarrow  2^{d^2} &\geq& c (N \cdot d^2) \nonumber
\end{eqnarray}
where $c = \frac{c_2}{c_1}$ is a constant.


We have empirically tested that for $d>3$, the rank based equivalence method is faster than its search based counterpart for bit flip error. 


The two algorithms based on linear search and  rank to determine gauge equivalence class representatives, are applicable for both bit flip and phase flip errors for any gauge code. However, for the heavy hexagonal code, we exploited the gauge structure to design a better algorithm to find gauge equivalence for phase flip errors. Similarly, structure of other codes may be exploited to design better algorithms to find gauge equivalence for bit and/or phase flip errors.




\section{Simulation Results}
\label{results}
 We have implemented and tested our ML based decoders for heavy hexagonal code \cite{chamberland2020topological} with bit flip, phase flip, and depolarization noise models. The machine learning parameters used for our decoder are given in Section 5.1. 

To the best of our knowledge, only MWPM based decoders have been studied for the heavy hexagonal QECC. Hence we have compared the metrics of our ML based decoder with those for MWPM decoder. We have used the stabiliser structure of the MWPM decoder proposed in  \cite{chamberland2020topological}, and implemented it using the `PyMatching' library \cite{higgott2022pymatching}. The threshold obtained by this implementation of MWPM decoder is the same as that in \cite{chamberland2020topological} (vide Fig.~\ref{fig:result3MWPMvsML}). Although the original paper on heavy hexagonal code used only MWPM decoder, for the sake of completeness, we have shown the threshold and pseudo-threshold values for Union Find decoder as well. We have used the Plaquette library \cite{plaquette} for the implementation of Union Find decoder.


We have executed  our codes for our ML decoders along with our search  and rank based algorithms to find the gauge equivalence representatives, on a server of National Energy Research Scientific Computing (NERSC) which has the computing resource Perlmutter, a Cray EX system with AMD EPYC CPUs and NVIDIA A100 GPUs.

\subsection{Machine Learning Parameters}

Our ML model has been trained in batches of data instead of training it with the entire data set (of size $10^6$) at once. This is advantageous in terms of both training time and memory capacity. For a distance $d$ heavy hexagonal code, there are $d^2$ data qubits and $d^2-1$ measure qubits (involving stabilizers and flag qubits). The data set is generated as follows. For each qubit (both data and measure), we apply an error with probability $p_{phys}$. If the noise model is bit (phase) flip, then a $X$ ($Z$) operator is applied on the particular qubit with probability $p_{phys}$. For depolarization noise model, the error operator is a linear combination of $X$, $Y$ or $Z$, each with probability $\frac{p_{phys}}{3}$. Finally, once the syndrome of the qubit error is generated (which is also noisy since the measure qubits may be erroneous as well), we apply measurement error by changing each bit of the syndrome with probability $p_{meas}$. According to \cite{fowler2012towards, chamberland2020topological}, we assume $p_{meas} = p_{phys}$. Under the assumption that errors are not correlated, this model captures a majority of the noise in real devices and has also been considered in prior research \cite{fowler2012towards, chamberland2020topological}. Thus our  dataset is a subset of all possible errors, which is likely to be a good representative of the actual noise in the hardware since it is generated in a manner similar to  the assumption for the errors to occur in the quantum hardware \cite{fowler2012towards, chamberland2020topological}.

The dataset size that we have used is 100000 from which 70000 is used for training and the rest for testing purpose. For evaluating the model we have used 10 fold cross validation. This approach involves randomly dividing the set of observations into 10 groups, or folds, of approximately equal size. The first fold is treated as a validation set, and the method is fit on the remaining 9 folds. 

The batch size of our simulation is 10000, run for 1000 epochs, with a learning rate of 0.01 (using stochastic gradient descent). The results have been averaged over five trials. This method is repeated for each value of the physical error probability ($p_{phys}$) considered in this study.  

Our FFNN has four layers, namely an input layer, an output layer and two hidden layers. For bit flip errors, in the input layer there are  $(d^2 -1)/2$ nodes which is the size of the syndrome (equals the number of $Z$ stabilizers). For phase flip errors, in the input layer there are  $d$ nodes which is the size of the syndrome (equals the number of $X$ stabilizers). The number of nodes in the output layer is always $d^2$, the number of data qubits for distance $d$ of the heavy hexagonal code. The loss function is chosen as binary cross-entropy. We have used the Keras library for creating the FFNN. Table \ref{tab:ffnntable} includes the details of our FFNN.

\begin{table}[htb]
    \centering
  
    \caption{ Number of nodes in the layers  of our FFNN based heavy hexagonal QECC syndrome decoder\\ for bit flip and phase flip errors }
    \begin{tabular}{|c|c|c|c|}
        \hline

          Layer & Code distance $d$ & \# nodes for bit flip errors & \# nodes for phase flip errors\\
         \hline
         &  3 & 4  & 3\\
        Input &  5 & 12 & 5\\
         &   7 &  24  & 7\\
         \hline
         & 3 & 9  & 9\\
          Output & 5 & 25  & 25\\
         & 7 & 49  & 49\\
         \hline
            & 3 & 16  & 16\\
      First hidden & 5 & 32 & 32\\
         & 7 & 64  & 64\\
         \hline
            & 3 & 16 & 16\\
        Second hidden & 5 & 32 &  32\\
         & 7 & 64  & 64\\
         \hline

    \end{tabular}
    
    \label{tab:ffnntable}
\end{table}

\newpage
\subsection{Estimating Logical Error Rate with our Machine Learning Decoder} 
The input of the FFNN based decoder is the syndrome or the list of ancilla measurements. The output of the decoder is the list of the indices of the data qubits for which error is detected. 
Suppose the size of the testing data is $N_t$. For each testing data $(S_i, e_i)$, where $S_i$ is the syndrome for the error $e_i$, the FFNN decoder is provided with the syndrome and it suggests a probable error $e_{i_{FFNN}}$. An error being a Pauli operator, is self-adjoint, and its correction is by simply applying the same operator once more. Therefore, the effective error on the system after performing correction based on the suggestion of the FFNN decoder is $e_{i_{FFNN}}.e_i$. Note that when $e_{i_{FFNN}} = e_i$ it implies perfect correction. We can verify easily whether $e_{i_{FFNN}}.e_i$ is a logical error. We want to emphasize here that the decoder may predict error incorrectly but the result may still not be a logical error, i.e., $e_{i_{FFNN}} \neq e_i$, whereas $e_{i_{FFNN}}.e_i$ is not a logical error. In accordance with \cite{chamberland2020topological}, we report here only the probability of logical errors identified as above. If this count of logical errors is $\nu_{logical}$, then our estimate for the logical rate or error probability $p_{logical}$ is reported as $\nu_{logical}/N_t$.

\subsection{Comparison of ML-based decoder results with MWPM}
\subsubsection{Bit flip error on data qubit}

First, we show the decoding performance of our FFNN-based decoder for heavy hexagonal codes of distances 3, 5, and 7 for bit flip noise model on data qubit (assuming ideal stabilizers and measurements). 

Table~\ref{tab:result1MWPMvsML} shows the values of threshold and pseudo-threshold for UF, MWPM and our proposed ML decoders for heavy hexagonal QECC with distances 3, 5, and 7. We note that the ML decoder clearly outperforms MWPM, which in its turn, outperforms UF. Therefore, henceforth, all comparisons of our ML decoder is performed with MWPM.
 
\begin{table}[htb]
    \caption{Comparison of UF, MWPM and FFNN based decoder result in case of bit flip error\\ (assuming ideal stabilizers and measurements)}
    \centering
    \begin{tabular}{|c|c|c|c|}
        \hline
        Decoding method & Code distance $d$ & Threshold &  Pseudo Threshold\\
      
        \hline
         & 3 &   & 0.0001\\
         UF & 5 & 0.0038 & 0.0008\\
         & 7 &  & 0.0018\\
      
         \hline
         & 3 &   & 0.0002\\
         MWPM & 5 & 0.0042 & 0.0012\\
         & 7 &  & 0.0024\\
         \hline
         & 3 &   & 0.006\\
         FFNN & 5 & 0.015 & 0.0086\\
         & 7 &  & 0.0115\\
         \hline

    \end{tabular}

    \label{tab:result1MWPMvsML}
\end{table}

\begin{figure}[htb]
    \centering
    \includegraphics[scale=0.4]{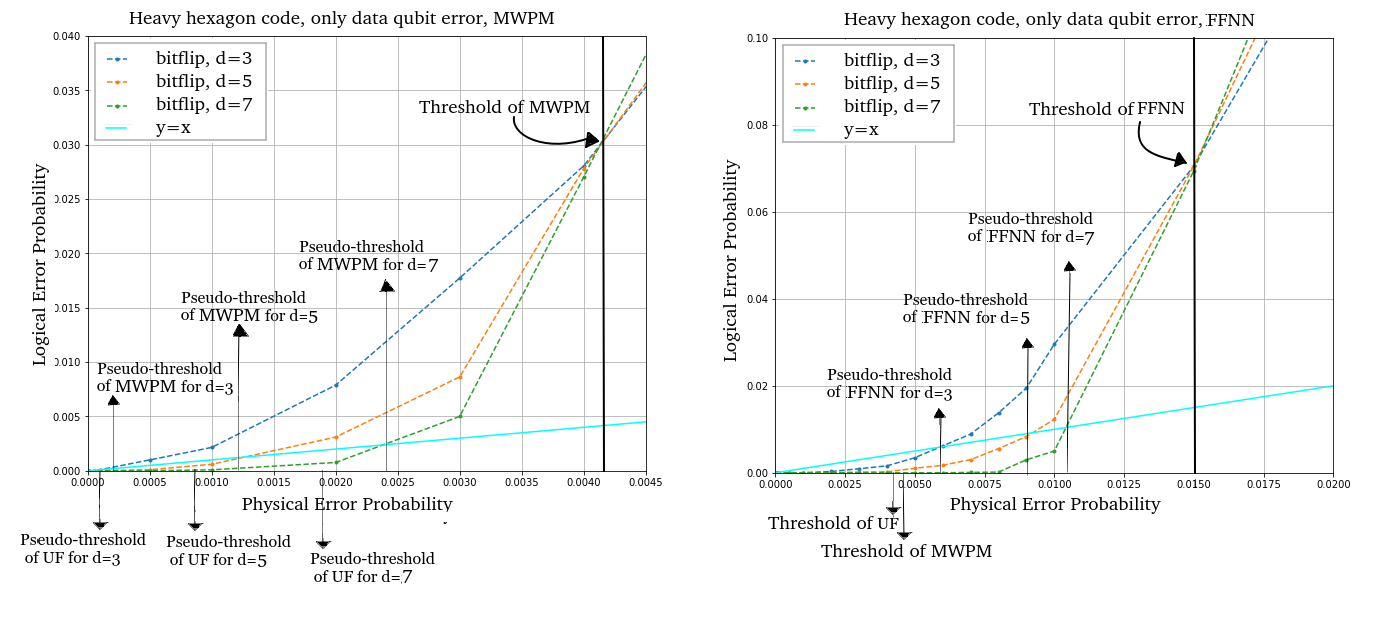}
    \caption{Threshold and pseudo-threshold values for MWPM  and ML based decoders in case of distance 3, 5, and 7 heavy hexagonal code for bit flip error on data qubit}
    \label{fig:result1MWPMvsML}
\end{figure}

Fig.~\ref{fig:result1MWPMvsML} shows the change in the probability of logical error with physical error probability $p$, which is the per-step error probability in the heavy hexagonal code cycle, when the error occurs only on data qubits. The results of MWPM and FFNN-based decoder for bit flip noise models on data qubits are shown. In the figure, the blue, orange, and green lines respectively are the plots for distances 3, 5, and 7. The points where the blue, orange and green line intersect the cyan straight line indicate the pseudo-threshold for distances 3, 5 and 7 respectively. 

Naturally, with increase in the distance of the heavy hexagonal code, the pseudo threshold improves. However, threshold is a property of the error correcting code and the noise model only. It is independent of the distance. In Fig. 6,  we observe that our ML based decoders are working much better than MWPM decoders for heavy hexagonal QECC in terms of threshold. 
 
\subsubsection{Performance of FFNN based decoder result without and with gauge equivalence for bit flip error on data qubit along with measurement and stabilizer error}
 
We show the decoding performance of our FFNN-based  decoder for heavy hexagonal codes of distances 3, 5, and 7 for bit flip noise model on data qubit along with measurement, and stabilizer error for three cases, i.e. without gauge equivalence, with search-based gauge equivalence, and with rank-based gauge equivalence. Our FFNN based model outperforms the existing MWPM decoder. We also show that the performance of FFNN with gauge equivalence is better than that of FFNN without gauge equivalence, and also the performance of search based  and rank based gauge equivalence is more or less equal. We then present the comparison of decoding time (in seconds) needed for distance 3, 5 and 7 heavy hexagonal code in case of FFNN based decoder for three cases, namely  without gauge equivalence, with search based gauge equivalence and with rank based gauge equivalence.  This supports the fact that rank based gauge equivalence method is faster than search based gauge equivalence  (Fig \ref{fig:timecompgauge}). Although the time needed for the basic FFNN decoder is much less than that for the search based gauge or rank based gauge equivalence methods but the decoding performance is poorer as already  mentioned above.
 
\begin{figure}[htb]
    \centering
    \includegraphics[scale=0.4]{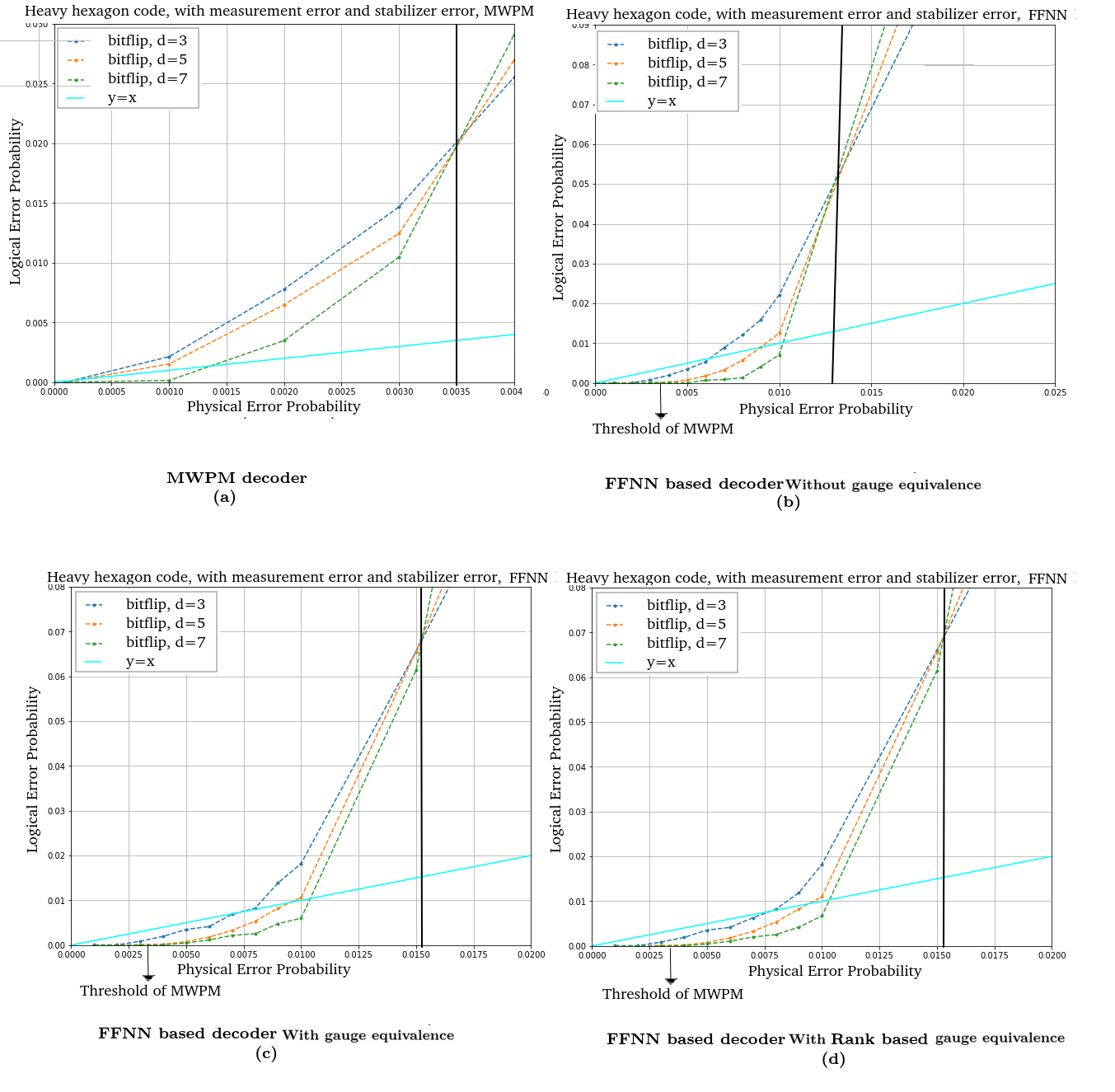}
    \caption{Threshold and pseudo-threshold values for MWPM and ML based decoders (without and with gauge equivalence) and  in case of distance 3, 5, and 7 heavy hexagonal code for bit flip error on data qubit along with measurement, and stabilizer error}
    \label{fig:result2MWPMvsML}
\end{figure}

Fig.~\ref{fig:result2MWPMvsML} shows the variation in the logical error probability with physical error probability $p$ with and without the application of gauge equivalence.

\begin{figure}[htb]
    \centering
    \includegraphics[scale=0.35]{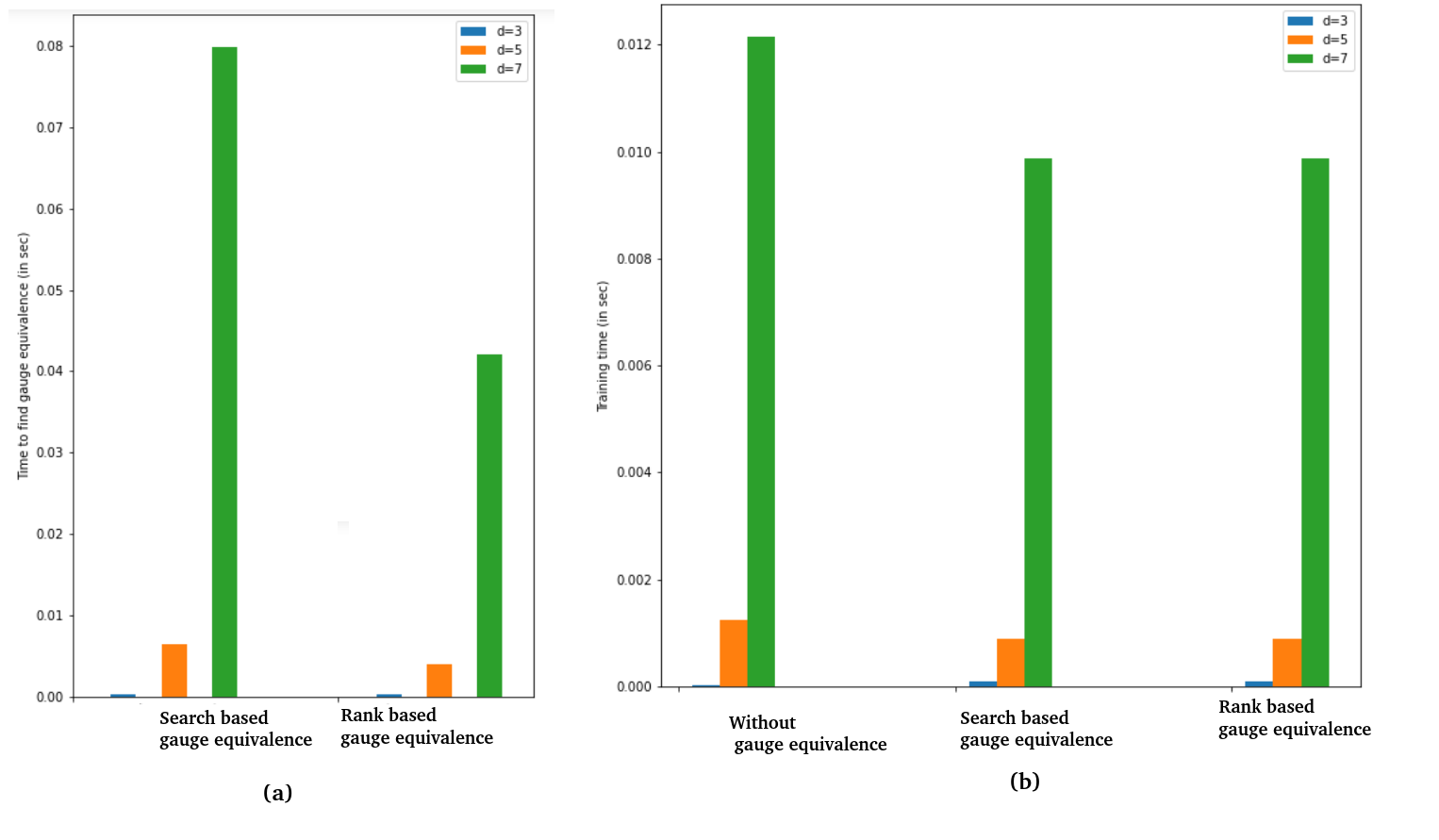}
    \caption{ Comparison of time (in seconds) needed to (a) pre-process and find gauge equivalence based on linear search  and on  rank for distance 3 (blue), distance 5 (orange)  and 7 (green) heavy hexagonal code, (b) train the FFNN based decoder for the three cases, namely without gauge equivalence, with search based gauge equivalence  and with rank based gauge equivalence}
    \label{fig:timecompgauge}
\end{figure}

\begin{table}[htb]
    \caption{ Comparison of FFNN based decoder, result without and with gauge equivalence in case of bit flip error on data qubit along with measurement, and stabilizer error  }
    \centering
    \begin{tabular}{|c|c|c|c|}
        \hline
         Decoding method & Code distance d & Threshold &   Pseudo Threshold\\
         \hline
         & 3 &  & 0.00018\\
         MWPM &5 & 0.0035  & 0.00023\\
         & 7 &  & 0.0013\\
         \hline
         & 3 &  & 0.0055\\
         FFNN without gauge equivalence & 5 & 0.01375 &  0.008\\
         & 7 &  & 0.011\\
         \hline
         & 3 &  & 0.0081\\
         FFNN with search based gauge equivalence & 5 & 0.01586  & 0.0102\\
         & 7 & & 0.0112\\
         \hline
         & 3 &  & 0.0082\\
         FFNN With rank based gauge equivalence & 5 & 0.01587  & 0.0103\\
         & 7 &   & 0.0112\\
         \hline

    \end{tabular}

    \label{tab:result2MWPMvsML}
\end{table}

In Table~\ref{tab:result2MWPMvsML} we note the threshold value for MWPM and FFNN based decoders (With and Without Using Gauge equivalence) and pseudo-threshold values for both decoders in case of distance 3, 5, and 7, for this same noise model.  Here we observe  $\sim 1.15\times $ (that is about $\sim 14\%$) increase in the threshold for ML-decoders using gauge equivalence as compared to  ML-decoders without using gauge equivalence. 

In Fig~\ref{fig:timecompgauge}(a) we compare the time to find the gauge equivalence. which is a preprocessing step before training the ML decoder. We observe that the rank  based method is faster than the search based method, which supports our Corollary 16 in Section 4.3.

In Fig~\ref{fig:timecompgauge}(b) we have shown the comparison of training time (in seconds) needed for distance 3, 5  and 7 heavy hexagonal code by our FFNN based decoder for three cases, namely without gauge equivalence, with search based gauge
equivalence and with rank based gauge equivalence. Decoding using gauge equivalence for both search  and rank based methods, takes less time than without gauge equivalence. The gauge equivalence technique reduces the number of error classes, resulting in a classification problem with fewer classes, and thus the training of ML model becomes not only faster but also more accurate.
For all the three cases, the time required for $d=7$ (green bar) is significantly higher than that for $d=5$ (orange bar) or $d=3$ (blue bar). 

For an instance of a decoding problem for $d=3$ heavy hexagon QECC using bit flip noise model,  while the decoding time for MWPM decoder is approx $6 \times 10^{-6} $ seconds, it is approx $3.8 \times 10^{-6} $ seconds for our FFNN based decoder using rank based gauge equivalence. This supports our intuition that ML based decoding is faster than MWPM based decoder for heavy hexagonal QECC.

\subsubsection{Performance of FFNN based decoder for phase flip error}  Fig. \ref{fig:phasegraph} depicts the performance of our FFNN-based decoder for heavy hexagonal code having distances 3, 5, and 7, using search based gauge equivalence for phase flip error on data qubit along with measurement and stabilizer errors. We observe that the probability of physical error above which the probability of logical error increases with the distance of the QECC, is $0.0118$ in the case of our FFNN based decoder, i.e., $\sim 19\times $ greater than that of MWPM based decoder, for which it is $0.00063$. For decoding phase flip errors in heavy hexagonal code, we have primarily made use of the Bacon-Shor code because as stated earlier heavy hexagonal code is a combination of surface code and subsystem code (Bacon Shor code) \cite{chamberland2018deep}. It may be pointed out that the definition of threshold as given in Section 1 is not applicable to Bacon-Shor codes \cite{bacon2006operator}.

\begin{figure}[htb]
    \centering
    \includegraphics[scale=0.3]{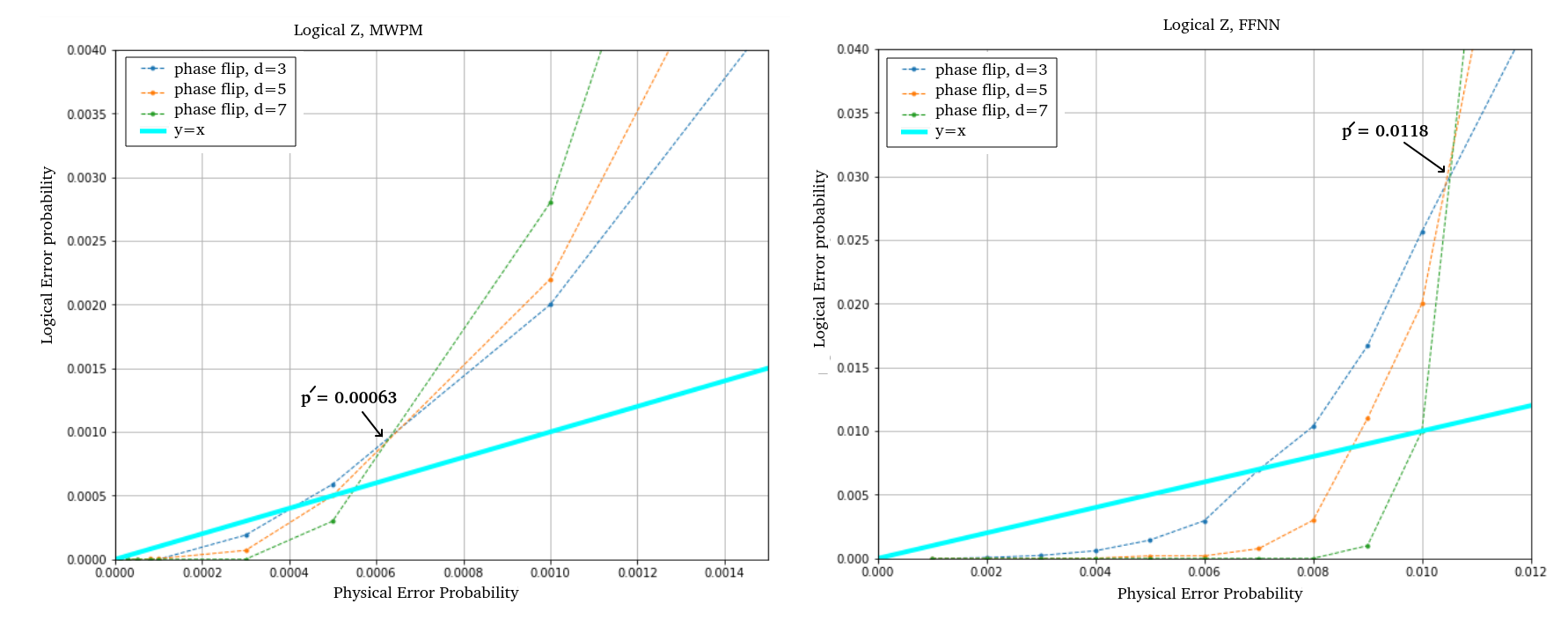}
    \caption{  Phase flip error rates for the heavy hexagonal code (a) In MWPM based decoder (b) In FFNN based decoder with gauge equivalence. }
    \label{fig:phasegraph}
\end{figure}

\subsubsection{Performance of FFNN based decoder for depolarization noise}

Lastly, we present the performance of our FFNN-based  decoder for heavy hexagonal codes of distances 3, 5, and 7 in the case of  depolarizing noise in terms of logical $X$ (bit flip) and logical $Z$ (phase flip) errors in Table~\ref{tab:depol}. Our model outperforms the existing decoders. 

\begin{table}[H]
 \caption{Comparison of FFNN based decoder with MWPM decoder for logical $X$  in depolarization noise model}
    \centering
    \begin{tabular}{|c|c|c|c|}
        \hline
         Decoding method & Code distance $d$  &  Threshold &  Pseudo-threshold\\
         \hline
         & 3 &   & 0.0005\\
         MWPM  & 5 & 0.0045 & 0.002\\
         & 7 &   & 0.0032\\
         \hline
         & 3 &  & 0.0105\\
         FFNN  ( with gauge equivalence) & 5 & 0.0245  & 0.0125\\
         & 7 &   & 0.0207\\
         \hline
    \end{tabular}
    \label{tab:depol}
\end{table}

Hence for depolarization noise model,  the threshold of our  FFNN based decoder is 0.0245 in the case of heavy hexagonal code,  which is much higher than that of MWPM, i.e., 0.0045 \cite{chamberland2020topological}. Hence our proposed ML-based decoding method achieves $\sim 5\times$ higher values of threshold than that by
MWPM. Further,  ML based decoders have better (higher) pseudo threshold  that MWPM.
Note: For depolarization noise model, the threshold for  FFNN based decoding was 0.03 with surface code \cite{bhoumik2022ml}. 

\begin{figure}[htb]
    \centering
    \includegraphics[scale=0.4]{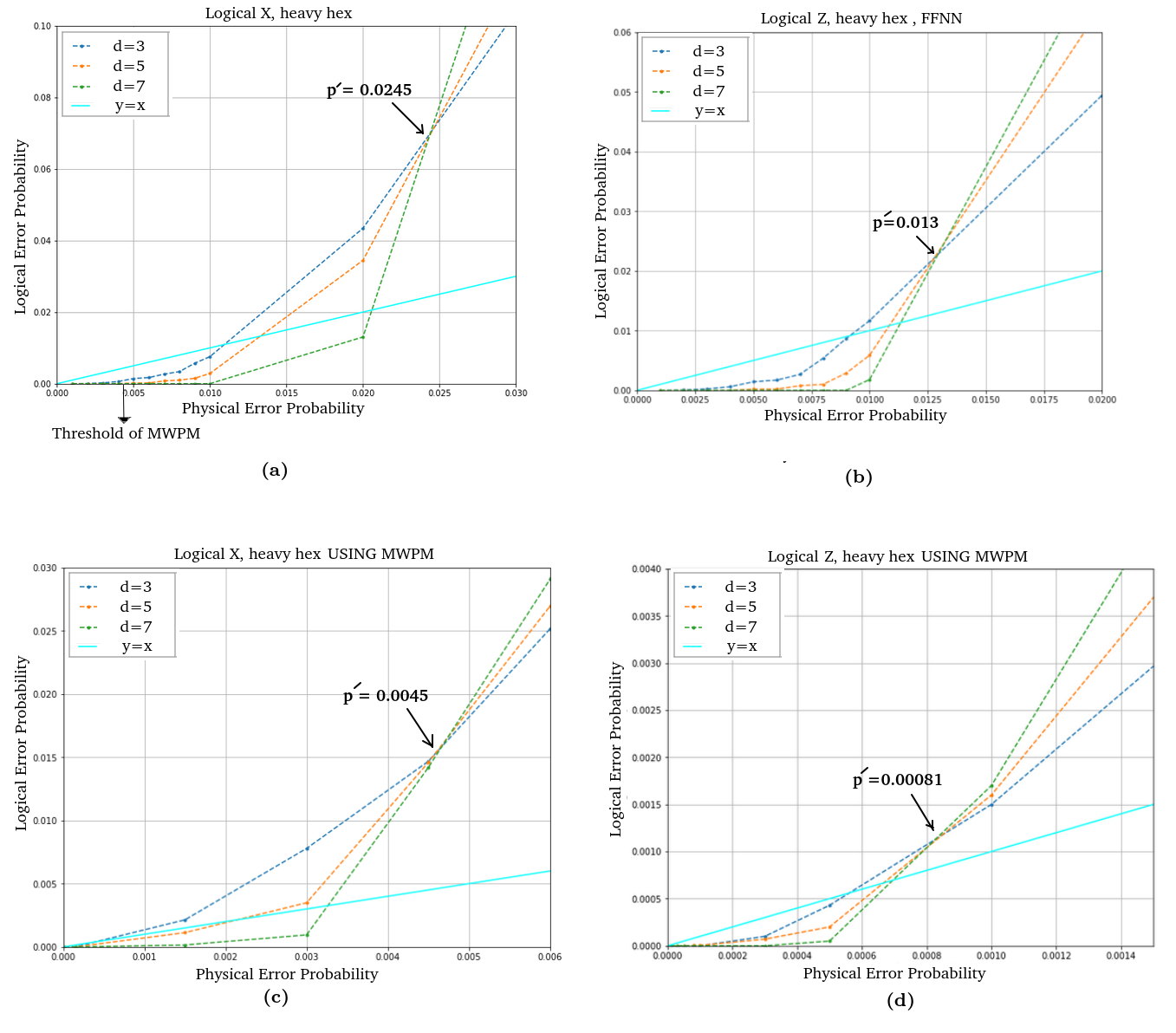}
    \caption{In FFNN based decoder for depolarizing noise model (a) Logical $X$ error rates and (b) logical $Z$ error rates for the heavy hexagonal code. In MWPM based decoder (c) Logical $X$ error rates and (d) logical $Z$ error rates for the heavy hexagonal code. }
    \label{fig:result3MWPMvsML}
\end{figure}

\subsection{Scalability of our ML based decoder}

A supervised ML model consists of a training and a testing phase. It is the training phase which often has a significant time requirement. For training a neural network having four layers with  $i$, $j$, $k$ and $l$ nodes respectively, $N$ training samples, and $n$ epochs, the worst case time complexity  is $\mathcal{O}(nN.(ij+jk+kl))$ \cite{alpaydin2021machine}. The number of nodes in each layer scales linearly with the number of qubits, which is $\mathcal{O}(d^2)$, $d$ being the distance of the QECC. The number of nodes for each layer is bounded by $\mathcal{O}(d^2)$. Therefore the time complexity for training is $\mathcal{O}(nN.d^4)$. While the training time of the ML model increases polynomially with the distance for a fixed-size set of training samples, it  increases linearly with the size of the training set for a fixed distance.

Once the training is performed, decoding of each syndrome is essentially executing the ML model once with that syndrome as the input. The time complexity for decoding each syndrome is, thus, $\mathcal{O}(ij+jk+kl) = \mathcal{O}(d^4)$.

Training of an ML model can be considered as a preprocessing step. With increase in distance of the code, it is necessary to increase the number of training samples $N$ in order to have a fair training. Since training can become expensive for higher distance codes, it may be possible to use a divide-and-conquer method to divide the lattice into multiple smaller (and possibly overlapping) sublattices so that each of them can be trained efficiently. However, this method has the challenge of \emph{knitting} the results from these sublattices into the final result corresponding to the lattice. We plan to explore this in future.

For graph based QECC decoders such as MWPM and UF, let $G = (V,E)$ be the syndrome graph, where $|V| = \mathcal{O}(d^2)$. The worst case time complexity of MWPM decoder is $\mathcal{O}(|V|^4) = \mathcal{O}(d^8)$ \cite{8880492}. Thus both the decoding time and the performance of MWPM is inferior to that of our proposed ML decoder. On the other hand, UF decoder scales almost linearly with the number of vertices on the graph as $\mathcal{O}(d^2\cdot \alpha(n))$ \cite{delfosse2021almost} where $\alpha(n)$ is the Ackermann's function. While this is lower than our ML decoder, it is to be noted that Union Find shows a significantly poor performance compared to both MWPM and our ML decoder in terms of error threshold which is the most important factor in QECC error syndrome decoding.

Therefore, the decoding time of our proposed ML decoder scales polynomially with the distance of the QECC, and performs significantly better  than both MWPM and UF decoder overall.

\section{Discussion}

In this paper, we have proposed an ML-decoder for heavy hexagonal code which makes use of a novel technique based on gauge equivalence to improve the performance of decoding. First we map decoding to a ML based classification problem. Exploiting the properties of heavy hexagonal code as a subsystem code, we have defined gauge equivalence, which, in turn, reduces the number of error classes for ML based classification. We have proposed \emph{search based} and \emph{rank based} methods for finding the gauge equivalence; the former being faster for finding equivalence in phase flip errors and the later for bit flip errors. 

We have tested our decoders for heavy hexagonal codes with distances upto 7 . We have shown that even the na\"ive ML decoder outperforms the MWPM and UF based decoders in terms of pseudo-threshold and threshold. Using gauge equivalence leads to further improvement in the  performance of our decoder. This research provides a plausible method for scaling current quantum devices to the fault-tolerant era. Future directions may be to study the applicability of gauge equivalence and ML based decoders for other noise models such as Pauli and amplitude damping noise.



\section*{Acknowledgments}
Debasmita Bhoumik would like to acknowledge fruitful discussions with Dr. Anupama Ray of IBM Research India. This research used resources of the National Energy Research Scientific Computing Center (NERSC), a U.S. Department of Energy Office of Science User Facility located at Lawrence Berkeley National Laboratory, operated under Contract No. DE-AC02-05CH11231 using NERSC award ERCAP0022238.

\appendix

\section{Appendix}

\subsection*{Enumerating Gauge Generators of Heavy Hexagonal QECC for $d=5$}
The 20 $Z$ gauge generators of the type $ Z_{i,j} Z_{i+1,j}$ in Fig.~\ref{fig:d5hexgauge} are given in Table \ref{table:appendix1}. The 8 $X$ gauge generators of the type $X_{i,j}X_{i,j+1}X_{i+1,j} X_{i+1,j+1}$  (represented by red squares in Fig.~\ref{fig:d5hexgauge}) are listed in Table \ref{table:appendix2}. 
The 2 $X$ gauge generators of the type $X_{1,2m-1} X_{1,2m}$ and 2 $X$ gauge generators of the type $X_{d,2m} X_{d,2m+1}$ (represented by red semicircles in Fig.~\ref{fig:d5hexgauge}) are  shown in Table \ref{table:appendix3}. Hence, there are a total of 12 $X$ gauge generators  in  Fig.~\ref{fig:d5hexgauge}.
In this appendix, we list all the gauge generators for a distance  heavy hexagonal code.  Here,
$i\in\{1,~2,~3,~4\}$, $j \in\{1,~2,~3,~4, ~5\}$ and $m \in \{1,~2\}$.
\begin{center}
\begin{table}[h!]
\caption{The 20 $Z$ gauge generators $Z_{i,j}Z_{i+1,j}$ with weight 2}
\begin{tabular}{ |c|c| } 
 \hline
 Gauge Generator & Position in lattice (blue semicircle) \\
  \hline
$Z_{1,1} Z_{2,1}$ & first row first column (top left)  \\
$Z_{2,1} Z_{3,1}$ & second row first column \\
... & \\ 
... & \\ 
$Z_{4,5} Z_{5,5}$ & last row last column\\ 
 \hline
\end{tabular}

\label{table:appendix1}
\end{table}
\end{center}

   
    

\begin{center}
\begin{table}[h!]
\caption{The 8 $X$ gauge generators $X_{i,j}X_{i,j+1}X_{i+1,j} X_{i+1,j+1}$ with weight 4}
\begin{tabular}{ |c|c| } 
 \hline
 Gauge Generator & Position in lattice ( red square) \\
  \hline
$X_{1,2}X_{1,3}X_{2,2} X_{2,3} $ &first row second column \\
$X_{1,4}X_{1,5}X_{2,4} X_{2,5} $ & first row fourth column\\
$X_{2,1}X_{2,2}X_{3,1} X_{3,2} $ & second row second column\\
$X_{2,3}X_{2,4}X_{3,3} X_{3,4} $ & second row fourth column\\
$X_{3,2}X_{3,3}X_{4,2} X_{4,3} $ & third row second column  \\ 
$X_{3,4}X_{3,5}X_{4,4} X_{4,5} $ & third row fourth column  \\ 
$X_{4,1}X_{4,2}X_{5,1} X_{5,2} $ & fourth row second column \\ 
$X_{4,3}X_{4,4}X_{5,3} X_{5,4} $ & fourth row fourth column\\
 \hline
\end{tabular}
\label{table:appendix2}
\end{table}
\end{center}


\begin{center}
\begin{table}[h!]
 \caption{The 4 $X$ gauge generators $X_{1,2m-1} X_{1,2m}$ and $~ X_{d,2m} X_{d,2m+1} $ with weight 2}
\begin{tabular}{ |c|c| } 
 \hline
 Gauge Generator & Position in lattice (red semicircle) \\
  \hline
$X_{1,1} X_{1,2}$ & first row first column (top left)  \\
$X_{1,3}X_{1,4}$ & first row third column \\

$X_{5,2}X_{5,3}$ & fifth row second column \\ 

$X_{5,4}X_{5,5}$ & fifth row fifth column \\ 

 \hline
\end{tabular}
\label{table:appendix3}
\end{table}
\end{center}
    


    

\subsection*{Enumerating Stabilizers of Heavy Hexagonal QECC for $d=5$}
The 8 $Z$ stabilizers of the type $ Z_{i,j} Z_{i,j+1} Z_{i+1,j} Z_{i+1,j+1}$ in Fig.~\ref{fig:d5hexgauge} are  shown in Table \ref{table:appendix4}. The 2 $Z$ stabilizers of the type  $Z_{2m-1,d} Z_{2m,d}$ and 2 $Z$ stabilizers of the type $Z_{2m,1} Z_{2m+1,1}$  (represented by blue semicircles in Fig.~\ref{fig:d5hexgauge}) are  given in Table \ref{table:appendix5}. Hence, there are a total of 12 $Z$ stabilizers in Fig.~\ref{fig:d5hexgauge}. Next, the 4 $X$ stabilizers of the type $\Pi_i X_{i,j} X_{i,j+1}$ in Fig.~\ref{fig:d5hexgauge} 
are shown in Table \ref{table:appendix6}.

\begin{center}
\begin{table}[h!]
\caption{The 8 $Z$ stabilizers $Z_{i,j} Z_{i,j+1} Z_{i+1,j} Z_{i+1,j+1}$ with weight 4}
\begin{tabular}{ |c|c| } 
 \hline
Stabilizer & Position in lattice (white square) \\
  \hline
$Z_{1,1} Z_{1,2} Z_{2,1} Z_{2,2}$ & first row first column \\
$Z_{1,3} Z_{1,4} Z_{2,3} Z_{2,4}$ & first row third column\\
$Z_{2,2} Z_{2,3} Z_{3,2} Z_{3,3}$ & second row second column   \\
$Z_{2,4} Z_{2,5} Z_{3,4} Z_{3,5}$ & second row fourth column  \\
$Z_{3,1} Z_{3,2} Z_{4,1} Z_{4,2}$ & third row first column  \\
$Z_{3,3} Z_{3,4} Z_{4,3} Z_{4,4}$ & third row third column  \\
$Z_{4,2} Z_{4,3} Z_{5,2} Z_{5,3}$ & fourth row second column \\
$Z_{4,4} Z_{4,5} Z_{5,4} Z_{5,5}$ & fourth row fourth column \\

 \hline
\end{tabular}
\label{table:appendix4}
\end{table}
\end{center}
     

\begin{center}
\begin{table}[h!]
 \caption{The 4 $Z$ stabilizers $~Z_{2m-1,d} Z_{2m,d}$ and $Z_{2m,1} Z_{2m+1,1}$ with weight 2}
\begin{tabular}{ |c|c| } 
 \hline
 Stabilizer & Position in lattice (blue semicircle) \\
  \hline
$Z_{1,5} Z_{2,5}$ & first row fifth column \\
$Z_{3,5} Z_{4,5}$ & third row fifth column \\ 
$Z_{2,1} Z_{3,1}$ & second row first column \\ 
$Z_{4,1} Z_{5,1}$ & fourth row first column \\ 

 \hline
\end{tabular}
\label{table:appendix5}
\end{table}
\end{center}







\begin{center}
\begin{table}[h!]
 \caption{The 4 $X$ stabilizers $\Pi_i X_{i,j} X_{i,j+1}$ }
\begin{tabular}{ |c|c| } 
 \hline
 Stabilizer & Position in lattice (vertical strip) \\
  \hline
$X_{1,1} X_{1,2} X_{2,1} X_{2,2} X_{3,1} X_{3,2} X_{4,1} X_{4,2} X_{5,1} X_{5,2}$ & the vertical strip of the column 1 and 2 \\
$X_{1,2} X_{1,3} X_{2,2} X_{2,3}  X_{3,2} X_{3,3}  X_{4,2} X_{4,3}  X_{5,2} X_{5,3} $ &  vertical strip of the column 2 and 3 \\ 
$X_{1,3} X_{1,4} X_{2,3} X_{2,4}  X_{3,3} X_{3,4}  X_{4,3} X_{4,4}  X_{5,3} X_{5,4} $ &  vertical strip of the column 3 and 4 \\
$X_{1,4} X_{1,5} X_{2,4} X_{2,5}  X_{3,4} X_{3,5}  X_{4,4} X_{4,5}  X_{5,4} X_{5,5} $ &  vertical strip of the column 4 and 5 \\
 \hline
\end{tabular}
\label{table:appendix6}
\end{table}
\end{center}


\bibliographystyle{unsrt}
\bibliography{apssamp}

\end{document}